\documentclass[11pt,peerreview,draftcls,onecolumn]{IEEEtran}

\newif\ifsubmission
\submissionfalse 

\usepackage[latin1]{inputenc}
\usepackage{amsmath}
\usepackage{amsfonts}
\usepackage{amssymb}
\usepackage{amstext}
\usepackage{amsthm}
\usepackage{graphicx}
\usepackage{url}
\usepackage{color}
\usepackage{bm}

\newtheorem{theorem}{Theorem}
\newtheorem{lemma}{Lemma}

\newtheorem{corollary}{Corollary}
\newtheorem{definition}{Definition}

\newcommand{\DD}{\mathcal D}
\newcommand{\XX}{\mathcal X}
\newcommand{\CC}{\mathcal C}
\renewcommand{\SS}{\mathcal S}

\newcommand{\PP}{\mathcal P}

\newcommand{\rarrow}{\tiny \to}

\graphicspath{{../images/}{images/}}
\ifsubmission
    \newcommand{\TODO}[1]{}
    
\else
    \newcommand{\TODO}[1]{\textcolor{red}{{#1}}}
    
\fi

\begin{document}

\title{Binary Hypothesis Testing Game with Training Data}


\author{Mauro Barni$^*$, \IEEEmembership{Fellow, IEEE}, Benedetta Tondi
\thanks{M. Barni is with the Department of Information Engineering, University of Siena, Via Roma 56, 53100 - Siena, ITALY, phone: +39 0577 234850 (int. 1005), e-mail: barni@dii.unisi.it; B. Tondi is with the Department of Information Engineering, University of Siena, Via Roma 56, 53100 - Siena, ITALY, e-mail: benedettatondi@gmail.com.}
}

\markboth{IEEE TRANSACTIONS ON INFORMATION THEORY}%
{M. Barni, B. Tondi: The Source Identification Game with Training Data}

\maketitle

\begin{abstract}
We introduce a game-theoretic framework to study the hypothesis testing problem, in the presence of an adversary aiming at preventing a correct decision. Specifically, the paper considers a scenario in which an analyst has to decide whether a test sequence has been drawn according to a probability mass function (pmf) $P_X$ or not. In turn, the goal of the adversary is to take a sequence generated according to a different pmf and modify it in such a way to induce a decision error. $P_X$ is known only through one or more training sequences. We derive the asymptotic equilibrium of the game under the assumption that the analyst relies only on first order statistics of the test sequence, and compute the asymptotic payoff of the game when the length of the test sequence tends to infinity. We introduce the concept of indistinguishability region, as the set of pmf's that can not be distinguished reliably from $P_X$ in the presence of attacks. Two different scenarios are considered: in the first one the analyst and the adversary share the same training sequence, in the second scenario, they rely on independent sequences. The obtained results are compared to a version of the game in which the pmf $P_X$ is perfectly known to the analyst and the adversary.
\end{abstract}

\begin{IEEEkeywords}
Hypothesis testing, adversarial signal processing, cybersecurity, game theory, source identification, multimedia forensics, counter-forensics.
\end{IEEEkeywords}

\IEEEpeerreviewmaketitle

\section{Introduction}

\IEEEPARstart{H}{pothesis} testing is a widely studied topic with applications in virtually all technological and scientific fields. In its most basic form, an analyst is asked to decide which among two hypotheses, usually referred to as null hypothesis (or $H_0$) and alternative hypothesis ($H_1$), is true based on a set of observables, say $x^n = (x_1, x_2 \dots x_n)$. Several versions of the problem are obtained according to the knowledge that the analyst has on the probability distribution of the observables when one of the two hypotheses holds. In some cases, the probability mass function (pmf)\footnote{Since in the rest of the paper we will assume that the elements of $x^n$ belong to a finite alphabet, we prefer to use the term probability mass function instead of probability density function, even if at this stage we do not need to restrict our attention to the discrete case.} conditioned to the two hypotheses is known, in other cases only the pmf under $H_0$ is known, in yet other cases only a number of sample observables (hereafter referred to as training sequences) obtained under $H_0$ and $H_1$ are available.

Due to its importance, hypothesis testing has been extensively studied and a solid theoretical framework has been built permitting to analyze and understand its many facets. In the last years, though, many applications have emerged in which hypothesis testing is given a new twist, due to the presence of an adversary aiming at making the test fail. In multimedia forensics \cite{SIMag09}, for instance, a forensic analyst may be asked to decide whether an image has been acquired by a given camera, notwithstanding the presence of an adversary aiming at deleting the acquisition traces left by the camera. In the same way, the analyst may be asked to decide whether a signal has undergone a certain processing or not, by taking into account the possibility that someone deliberately tried to delete the traces left during the processing phase.

Another popular example comes from spam filtering \cite{DDMSV04}, wherein an anti-spam filter is presented with a test e-mail and must decide whether the e-mail contains a genuine or a spam message. It is obvious that such a test can not neglect the presence of an adversary trying to shape the message in such a way to fool the filter.

Biometric authentication provides a further example. In this case the authentication system must decide whether a biometric template belongs to a certain individual, despite the opposite efforts of an attacker aiming at building a fake template that passes the authentication test \cite{JRU05,MDiaz06}.

Other examples include: watermarking, where the detector is asked to decide whether a document contains a given watermark or not \cite{CoxWat02}, steganalysis, in which the steganalyzer has to distinguish between cover and stego images \cite{Fridrich09}, network intrusion detection, wherein anomalous traffic conditions must be distinguished from normal ones \cite{TKM02}, reputation systems \cite{YSKY09}, for which it is essential to distinguish between genuine and malevolent scores, anomaly detection, cognitive radio \cite{WLSH10}, and many others. In all these fields, the system designer has to take into account the presence of one or more adversaries explicitly aiming at system failure.


In the framework depicted above, the goal of this paper is to move a first step towards the construction of a general theoretical framework to analyze the binary hypothesis testing problem by taking into account the presence of an adversary aiming at impeding a correct decision. More specifically, we introduce and analyze an adversarial version of the Neyman-Pearson setup in which an analyst and an adversary, hereafter referred to as the Defender (D) and the Attacker (A), face each other in a rigorously defined context. As in the classical Neyman-Pearson scenario, we assume that the type-I error probability (i.e., the probability of rejecting $H_0$ when $H_0$ holds) is {\em in some way} fixed and that D and A are interested in minimizing, res. maximizing, the type-II error probability (i.e. the probability that the analyst accepts $H_0$ when $H_1$ holds). In order to do so, we adopt a game-theoretic framework, in which the defender and the attacker have opposite goals and operate by satisfying a different set of requirements, all together specifying the nature of the game. The final goal will be the derivation of the optimum strategies for the defender and the attacker in terms of game equilibrium points, and the study of the achievable performance at the equilibrium.

\subsection{Prior art}
The use of game theory to model the impact that the presence of an adversary has on (binary) hypothesis testing is not an absolute novelty. In many security oriented fields in which hypothesis testing plays a central role, game theory has been advocated to avoid entering a cat and mouse loop in which researchers alternatively play the role of the defender and the attacker, and continuously develop new countermeasures, each time by attacking a specific algorithm or strategy.

By referring to multimedia forensics, in  \cite{Boh12} B{\"o}hme and Kirchner cast the forensic problem in a hypothesis testing framework. Several versions of the problem are defined according to the particular hypothesis being tested, including distinction between natural and computer generated images, manipulation detection, source identification. Counter-forensics is then defined as a way to degrade the performance of the hypothesis test envisaged by the analyst. By relying on arguments similar to those used in steganography and steganalysis \cite{Cach98}, B{\"o}hme and Kirchner argue that the divergence between the probability density functions of the observed signals after the application of the counter-forensic attack is a proper way to measure the effectiveness of the attack. Noticeably, such measure does not depend on the particular investigation technique adopted by the analyst. Even if B{\"o}hme and Kirchner do not explicitly use a game-theoretic formulation, their attempt to decouple the counter-attack from a specific forensic strategy can be seen as a first - implicit - step towards the definition of the equilibrium point of a general multimedia forensics game. Another work loosely related to the present paper is \cite{Stamm12}, where the authors introduce a game-theoretic framework to evaluate the effectiveness of a given attacking strategy and derive the optimum countermeasures. As opposed to our analysis, in \cite{Stamm12} the attacker's strategy is fixed and the game-theoretic framework is used only to determine the optimum parameters of the forensic analysis and the attack, thus failing to provide a complete characterization of the game between the attacker and the analyst.

Game theory and information theory have been used in watermarking to model the interplay between the watermaker and the attacker. In  \cite{CL2002, Moulin04, SMM04}, the game is played between the watermark embedder/decoder and an attacker who attempts to degrade the embedded message by modifying the watermarked signal, e.g. by adding some noise. The payoff of the game is usually the capacity of the watermark channel. A problem that is closer to the one addressed in this paper is the one considered in \cite{merhav_sabbag_08}, where the jointly optimum embedding and detection strategies for a detector with limited resources are derived. Indeed, the approach used in the present paper is reminiscent of the analysis carried out in \cite{merhav_sabbag_08}, given that in both cases the analysis focuses on an asymptotic version of the problem in which the resources available to the defender are limited. As opposed to the present work, however, the analysis in \cite{merhav_sabbag_08} is carried out under the assumption that no attack is present or that the attack channel is fixed, and the resort to a min-max optimization is due only to the necessity of finding the jointly optimum watermark embedding and detection strategies.

The work that is most closely related to the present paper is \cite{BT13}, where the source identification game with known statistics is introduced and the corresponding asymptotic Nash equilibrium derived. As a matter of fact, even if \cite{BT13} restricts the analysis to multimedia forensics, the framework adopted to model the game between the forensic analyst and the adversary is very general and can also be used to model a binary hypothesis testing problem in which the statistics of the observables under the null hypothesis are perfectly known to the defender and the attacker. In order to avoid replicating the analysis carried out in \cite{BT13}, by incorporating only the few  modifications needed to extend it from multimedia forensics to general hypothesis testing, here we focus on a different version of the game in which the statistics of the observables under $H_0$ are known only through training data. This represents a major deviation from \cite{BT13}, requiring a thorough reformulation of the problem and a new derivation of the equilibrium point. In order to make the current paper self-contained and allow the reader to better appreciate the difference between the results obtained in \cite{BT13} and the new findings provided here, we summarize the main definitions and results of \cite{BT13} in Section \ref{sec.KS}.

\subsection{Contribution}

With the above ideas in mind, in this paper we address the following problem, hereafter referred to as the binary hypothesis testing problem with training data ($HT_{tr}$). Let  $X \sim P_X$ be a discrete memoryless (DM) source ruling the emission of observables under $H_0$, and let $x^n$ be a test sequence, i.e., a sequence of observables. The goal of the defender is to accept or reject hypothesis $H_0$, that is, to decide whether $x^n$ was drawn from $X$ or not\footnote{In order to keep the notation as light as possible, we use the symbol $x^n$ to indicate the test sequence even if, in principle, it is not known whether $x^n$ originated from $X$ or not.}. In doing so, D must ensure that the type-I error probability does not exceed a predefined value. Let then $Y \sim P_Y$ be a second DM source and let $y^n = (y_1, y_2 \dots y_n)$ be a sequence generated by $Y$. It is the aim of A to transform $y^n$ into a new sequence $z^n$ in such a way that when presented with $z^n$ the defender accepts $H_0$. We also impose that $A$ satisfies a distortion constraint requiring that the {\em distance} between $y^n$ and $z^n$ is below a certain threshold. As opposed to \cite{BT13}, D and A do not know the exact statistics of $P_X$ and $P_Y$, since all they know is a training sequence drawn from $X$.

Given the above scenario, the goal of this paper is to propose a rigorous game-theoretic framework to cast the $HT_{tr}$ problem in, and derive the asymptotic equilibrium point of the game under a simplifying hypothesis about the kind of analysis which D can carry out. We will do so for two different versions of the game stemming from different assumptions about the relationship between the training sequence available to the defender and that available to the attacker. In a first case, we will assume that A and D share the same training sequence, while in the second part of the paper we will assume that the two sequences are generated independently from each other. The main results proven in the paper can be summarized as follows:
\begin{enumerate}
\item{We show that under the {\em limited resources} assumptions \cite{merhav_sabbag_08}, the $HT_{tr}$ game admits an asymptotic equilibrium. We also prove that the asymptotic equilibrium point is the only rationalizable equilibrium of the game \cite{Bern84, Pea84}. Such an equilibrium is much stronger than the usual notion of Nash equilibrium, since the strategies corresponding to such an equilibrium are the only ones that two rationale players may adopt (Theorem \ref{theo.Nash}, Section \ref{sec.Equilibrium_tr_c})};
\item{We compute the payoff at the equilibrium for the defender and the attacker, and introduce the notion of indistinguishability region, defined as the region with the $P_Y$'s that can not be distinguished reliably (i.e. with a vanishing type II error probability) from $P_X$ (Theorem \ref{theo.SanovTRc}, Section \ref{sec.payoff});}
\item{We compare the achievable payoff of the $HT_{tr}$ game with the results obtained in \cite{BT13}, where D and A have a perfect knowledge of $P_X$, showing that the $HT_{tr}$ game is more favorable to the attacker with respect to the situation analyzed in \cite{BT13} (Theorem \ref{theo.inclusion}, Section \ref{sec.payoff});}
\item{We show that the indistinguishability region is the same when D and A share the same training sequence and when they rely on independent sequences (Theorem \ref{theo.HTa}, Section \ref{sec.version_a}).}
\end{enumerate}
With regard to 1), the asymptotic equilibrium point when D and A rely on the same training sequence was already derived in \cite{BT12}, without realizing that the equilibrium is indeed stronger than a Nash equilibrium. The case of independent training sequences has never been studied before. The methodology used to derive the equilibrium point goes along the same lines used in \cite{merhav_sabbag_08} to derive the jointly optimum watermark embedding and detection strategy. Finally, the optimum strategy of the defender can be paralleled to the results obtained in \cite{Gut89}, even if in a completely different context. As to point 2), our results are closely related to Sanov's theorem \cite{CandT, CSnow}, however, to the best of our knowledge, their derivation in a game-theoretic context is not trivial and represents an original contribution of the present paper.

The rest of this work is organized as follows. In Section \ref{sec.not} we introduce the notation and definitions used throughout the paper. In Section \ref{sec.KS}, we recall the main results proved in \cite{BT13} by casting them into a hypothesis testing framework. Such results represent the baseline against which we will compare the new results obtained in this paper. In Section \ref{sec.games}, we formally introduce the $HT_{tr}$ game and lay the basis for the analysis carried out in the subsequent sections. Section \ref{sec.Equilibrium_tr_c} is devoted to the derivation of the equilibrium point of the $HT_{tr}$ game when A and D share the same training sequence. The payoff at the equilibrium is analyzed in Section \ref{sec.payoff}, where we also introduce the notion of indistinguishability region. The case of independent training sequences is analyzed in Section \ref{sec.version_a}. The paper ends in Section \ref{sec.conc}, with some conclusions and hints for future research. Some of the most technical proofs are given in the the appendix, so to avoid interrupting the main flow of ideas.

\section{Notation and definitions}
\label{sec.not}

In the rest of this work we will use capital letters to indicate discrete memoryless sources (e.g. $X$). Sequences of length $n$ drawn from a source will be indicated with the corresponding lowercase letters (e.g. $x^n$). In the same way, we will indicate with $x_i$, $i =1, n$ the $i-$th element of a sequence $x^n$. The alphabet of an information source will be indicated by the corresponding calligraphic capital letter (e.g. $\XX$). Calligraphic letters will also be used to indicate classes of information sources ($\CC$). The pmf of a discrete memoryless source $X$ will be denoted by $P_X$. The same notation will be used to indicate the probability measure ruling the emission of sequences from a source $X$, so we will use the expressions $P_X(a)$ and $P_X(x^n)$ to indicate, respectively, the probability of symbol $a \in \XX$ and the probability that the source $X$ emits the sequence $x^n$, the exact meaning of $P_X$ being always clearly recoverable from the context wherein it is used. Given an event $A$ (be it a subset of $\XX$ or $\XX^n$), we will use the notation $P_X(A)$ to indicate the probability of the event $A$ under the probability measure $P_X$.

Our analysis relies extensively on the concepts of type and type class defined as follows (see \cite{CandT} and \cite{CandK} for more details). Let $x^n$ be a sequence with elements belonging to an alphabet $\XX$. The type $P_{x^n}$ of $x^n$ is the empirical pmf induced by the sequence $x^n$, i.e. $\forall a \in \XX, P_{x^n} (a) = \frac{1}{n} \sum_{i=1}^n \delta(x_i, a)$, where $\delta(x_i,a) = 1$ if $x_i =a$ and zero otherwise. In the following we indicate with $\PP_n$ the set of types with denominator $n$, i.e. the set of types induced by sequences of length $n$. Given $P \in \PP_n$, we indicate with $T(P)$ the type class of $P$, i.e. the set of all the sequences in $\XX^n$ having type $P$.

The Kullback-Leibler (KL) divergence between two distributions $P$ and $Q$ on the same finite alphabet $\XX$ is defined as:
\begin{equation}
    \DD(P||Q) = \sum_{a \in \XX} P(a) \log \frac{P(a)}{Q(a)},
\label{eq.KLdiv}
\end{equation}
where, as usual, $0 \log 0 = 0$ and $p \log p/0 = \infty$ if $p > 0$.

\subsection{Hypothesis testing framework}

Given a sequence $x^n \in \XX^n$, as a result of the test, $\XX^n$ is partitioned into two complementary regions $\Lambda$ and $\Lambda^c$, such that for $x^n \in \Lambda$ the defender decides in favor of $H_0$, while for $x^n \in \Lambda^c$ $H_1$ is preferred. We say that a Type-I error occurs if $H_1$ is chosen even if $H_0$ holds. In the same way, we say that a Type-II error occurs when $H_1$ holds but $H_0$ is chosen. In the following, we will refer to Type-I errors as false positive errors (or false alarms) and to Type-II as false negative (or missed detection), and will indicate the probability of such events as $P_{fp}$ and $P_{fn}$ respectively. The motivation for such a terminology comes from applications in which $H_0$ is seen as a standard situation and its rejection in favor of $H_1$ raises an alarm since something unusual happened. It goes without saying that our derivation remains valid even in different scenarios where the false positive and false negative terms may not be appropriate. In our analysis we are mainly interested in the asymptotic behavior of $P_{fp}$ and $P_{fn}$ when $n$ tends to infinity. In particular we define the false positive ($\lambda$) and false negative ($\varepsilon$) error exponents as follows:
\begin{equation}
    \lambda = \lim_{n \rarrow \infty} -\frac{\log P_{fp}}{n}; ~~~~\varepsilon = \lim_{n \rarrow \infty} -\frac{\log P_{fn}}{n},
\label{eq.err_exp}
\end{equation}
where the $\log$'s are taken in base 2.

\subsection{Game Theory}

As we said, the goal of this paper is to model the  binary hypothesis testing problem in an adversarial setting as a 2-player game. More formally, a 2-player game is defined as a 4-uple $G(\SS_1,\SS_2,u_1, u_2)$, where $\SS_1 = \{s_{1,1} \dots s_{1,n_1}\}$ and $\SS_2 = \{s_{2,1} \dots s_{2,n_2}\}$ are the set of strategies (actions) the first and the second player can choose from, and $u_l(s_{1,i}, s_{2,j}), l= 1,2$, is the payoff of the game for player $l$, when the first player chooses the strategy $s_{1,i}$ and the second chooses $s_{2,j}$. A pair of strategies $(s_{1,i}, s_{2,j})$ is called a profile. In a zero-sum competitive game, the two payoff functions are strictly related to each other since for any profile we have $u_1(s_{1,i}, s_{2,j}) + u_2(s_{1,i}, s_{2,j}) = 0$. In other words, the win of a player is equal to the loss of the other. In the particular case of a zero-sum game, then, only one payoff function needs to be defined. Without loss of generality we can specify the payoff of the first player (generally indicated by $u$), with the understanding that the payoff of the second player $u_2$ is equal to $-u$. In the most common formulation, the sets $\SS_1$, $\SS_2$ and the payoff functions are assumed to be known to both players. In addition, it is assumed that the players choose their strategies before starting the game so that they have no hints about the  strategy actually chosen by the other player (strategic game).

A common goal in game theory is to determine the existence of equilibrium points, i.e. profiles that, in {\em some sense} represent a {\em satisfactory} choice for both players. While there are many definitions of equilibrium, the most famous and commonly adopted is the one due by Nash \cite{Nash50,Osb94}. For the particular case of a 2-player game, a profile $(s_{1,i^*}, s_{2,j^*})$ is a Nash equilibrium if:
\begin{equation}
\begin{array}{ll}
    u_1((s_{1,i^*}, s_{2,j^*})) \ge u_1((s_{1,i}, s_{2,j^*})) & \forall s_{1,i} \in \SS_1\\
    u_2((s_{1,i^*}, s_{2,j^*})) \ge u_2((s_{1,i^*}, s_{2,j})) & \forall s_{2,j} \in \SS_2,
\end{array}
\label{eq.Nash}
\end{equation}
where for a zero-sum game $u_2 = -u_1$. In practice, a profile is a Nash equilibrium if each player does not have any interest in changing its choice assuming the other does not change its strategy.

Despite its popularity, the practical meaning of Nash equilibrium points is difficult to grasp, since there is no guarantee that the players will end up playing at the equilibrium. This is particularly evident when more than one Nash equilibrium exists. A definition of game equilibrium with a more practical meaning can be obtained by relying on the notion of dominant and dominated strategies. A strategy is said to be strictly dominant for one player if it is the best strategy for the player, no matter how the other player may play. In a similar way, we say that a strategy $s_{l,i}$ is strictly dominated by strategy $s_{l,j}$, if the payoff achieved by player $l$ choosing $s_{l,i}$ is always lower than that obtained by playing $s_{l,j}$ regardless of the choice made by the other player. The recursive elimination of dominated strategies is one common technique for solving games. In the first step, all the dominated strategies are removed from the set of available strategies, since no rational player would ever play them. In this way a new smaller game is obtained. At this point, some strategies, that were not dominated before, may be dominated in the remaining  game, and hence are eliminated. The process goes on until no dominated strategy exists for any player. A {\em rationalizable equilibrium} is any profile which survives the iterated elimination of dominated strategies. If at the end of the process only one profile is left, the remaining profile is said to be the {\em only rationalizable equilibrium} of the game, which is also the only Nash equilibrium point. The corresponding strategies are the only rational choice for the two players and we say that the game is {\em dominance solvable}. Dominance solvable games are easy to analyze since, under the assumption of rational players, we can anticipate that the players will choose the strategies corresponding to the unique rationalizable equilibrium \cite{ChenGames}.
\section{Binary hypothesis testing game \\ with known sources}
\label{sec.KS}

In this section, we use the framework introduced in \cite{BT13} to define a version of the hypothesis testing game in which the pmf's ruling the emission of sequences from $X$ and $Y$ are known to both D and A. We also summarize the main results proven in \cite{BT13}, so to ease the comparison with the new results that will be proven in the rest of the paper.

Let $X$ and $Y$ be two DM sources with the same alphabet $\XX$. Let $y^n$ be a sequence drawn from $Y$ and let $z^n$ be a modified version of $y^n$ produced by A in the attempt to deceive D. The binary hypothesis testing game under the known source assumption ($HT_{ks}$) is defined as follows.

\begin{definition}
The $HT_{ks}(\SS_{D}, \SS_{A}, u)$ game is a zero-sum, strategic, game played by D and A, defined by the following strategies and payoff.
\begin{itemize}
\item{The set of strategies D can choose from is the set of acceptance regions $\Lambda$ for which the false positive probability is below a certain threshold:
\begin{equation}
    \SS_{D} = \{ \Lambda : P_X(x^n \notin \Lambda) \le P_{fp}\}.
\label{eq.SD_KS}
\end{equation}}
\item{The set of strategies A can choose from is formed by all the functions that map a sequence $y^n \in \XX^n$ into a new sequence $z^n \in \XX^n$ subject to a distortion constraint:
\begin{equation}
    \SS_{A} = \{ f(y^n): d(y^n, z^n) \le nD\},
\label{eq.SA_KS}
\end{equation}
where $d(\cdot,\cdot)$ is a proper distance measure and $D$ is the maximum allowed per-letter distortion.}
\item{The payoff function is defined as the false negative error probability ($P_{fn}$), namely:
\begin{equation}
    u(\Lambda, f) = - P_{fn} =  - \sum_{y^n: f(y^n) \in \Lambda} P_Y(y^n).
\label{eq.payoff_KS}
\end{equation}
}
\end{itemize}
\label{def.HT_ks}
\end{definition}

Given the difficulty of studying the game defined above, a simplified version of the game is introduced in \cite{BT13} wherein the set of strategies available to D is limited. More specifically, the so-called limited resources assumption is introduced, forcing D to base its analysis only on first order statistics of $x^n$. Stated in another way, it is required that $\Lambda$ is a union of type classes. Since a type class is univocally defined by the empirical probability mass function of the sequences contained in it, the acceptance region $\Lambda$ can be seen as a union of types $P \in \PP_n$. As an additional simplification, the constraint on the false positive probability is defined in asymptotic terms, requiring that $P_{fp}$ decreases exponentially fast with a given decay rate. All these considerations lead the following definition:

\begin{definition}
The $HT_{ks}^{lr}(\SS_{D}, \SS_{A}, u)$ game is a game between D and A defined by the following strategies and payoff:
\begin{equation}
    \SS_{D} = \{ \Lambda \in 2^{\PP_n}: P_{fp} \le 2^{-\lambda n}\},
\label{eq.SD_KS_as}
\end{equation}
\begin{equation}
    \SS_{A} = \{ f(y^n): d(y^n, f(y^n)) \le nD\},
\label{eq.SA_KS_as}
\end{equation}
\begin{equation}
    u(\Lambda, f) = -P_{fn},
\label{eq.payoff_KS_as}
\end{equation}
\label{def.SI_ks_as}
\end{definition}
where $2^{\PP_n}$ indicates the power set of $\PP_n$.
From the analysis given in \cite{BT13} we know the following results.
\begin{theorem}
The profile $(\Lambda_{ks}^*, f_{ks}^*)$ with
\begin{equation}
    \Lambda_{ks}^* = \left\{P \in \PP_n: \DD(P || P_X) < \lambda - |\XX| \frac{\log(n+1)}{n} \right\},
\label{eq.opt_lambda_KS_as}
\end{equation}
and
\begin{equation}
    f_{ks}^*(y^n) = \arg\min_{z^n : d(z^n, y^m) \le nD} \DD(P_{z^n} || P_X).
\label{eq.opt_AD_KS_as}
\end{equation}
defines an asymptotic Nash equilibrium for the $HT_{ks}^{lr}$ game.
\label{theo.NASH_KS_as}
\end{theorem}
As a matter of fact, from the proof given in \cite{BT13}, it is easy to see that $(\Lambda_{ks}^*, f_{ks}^*)$ is the only rationalizable equilibrium of the game and hence $HT_{ks}^{lr}$ is a dominance solvable game.

A fundamental consequence of Theorem \ref{theo.NASH_KS_as} is that the optimum strategies for D and A do not depend on $P_Y$ hence making the assumption that $P_Y$ is known irrelevant. With a few modifications, then, Theorem \ref{theo.NASH_KS_as}  can be applied to a composite hypothesis testing scenario in which only the pmf conditioned to $H_0$ is known \cite{Kay}.

The second main result proven in \cite{BT13} regards the payoff at the equilibrium, and specifies under which conditions it is possible for D to devise a decision strategy such that $P_{fn}$ tends to zero exponentially fast when $n$ tends to infinity. Let  $\Gamma_{ks}^n$ be defined as follows:
\begin{equation}
\Gamma^n_{ks} \hspace{-3pt} = \hspace{-3pt} \{ P \in \PP_n : \forall y^n \in T(P), \exists z^n \in \Lambda_{ks}^* \text{ s.t. } d(y^n, z^n) \le nD \}
\label{eq.Gamman_KS}
\end{equation}
and let the asymptotic version of $\Gamma_{ks}^{n}$ be defined as
\begin{equation}
    \Gamma^{\infty}_{ks} = \text{{\em cl}} \left( \bigcup_n \Gamma^n_{ks} \right),
\label{eq.Gamma_union}
\end{equation}
where $\text{{\em cl}}(S)$ indicates the {\em closure} of the set $S$. The following theorem holds:
\begin{theorem}
For the $HT_{ks}^{lr}$ game, the error exponent of the false negative error probability at the equilibrium is given by:
\begin{equation}
    \varepsilon_{ks} = \min_{P \in \Gamma_{ks}^{\infty}} \DD(P||P_Y),
\label{eq.asymptotic_theorem}
\end{equation}
leading to the following cases:
\begin{enumerate}
    \item{$\varepsilon_{ks} = 0$, if $P_Y \in \Gamma_{ks}^{\infty}$;}
    \item{${\displaystyle \varepsilon_{ks} \ne 0}$, if $P_Y \notin \Gamma_{ks}^{\infty}$}.
\end{enumerate}
\label{theo.payoff_KS_as}
\end{theorem}
where $\varepsilon_{ks}$ indicates the false negative error exponent at the equilibrium. Given two pmf's $P_X$ and $P_Y$, a distortion constraint $D$ and the desired false positive error exponent $\lambda$, Theorem \ref{theo.payoff_KS_as} permits to understand whether D may ever succeed to make the false negative error probability vanishingly small and thus {\em win} the game. As a matter of fact, this is possible only if $P_Y \notin \Gamma_{ks}^{\infty}$, since otherwise $\varepsilon_{ks} = 0$. We will call $\Gamma_{ks}^{\infty}$ the indistinguishability region for $P_X$, i.e. set of sources that under certain conditions (summarized by $D$ and $\lambda$) are not distinguishable from $P_X$.

\section{Binary hypothesis testing game \\ with training data}
\label{sec.games}

The analysis carried out in \cite{BT13} requires that $P_X$ and $P_Y$ are known to D and A (as we have seen in the previous section, in the asymptotic case only the knowledge of $P_X$ is required). To get closer to a realistic scenario, we now remove this assumption introducing the hypothesis testing game with training data.

Let $\CC$ be the class of discrete memoryless sources with alphabet $\XX$, and let $X \simeq P_X$ be a source in $\CC$. As for the $HT_{ks}$ game, the goal of D is to decide whether a test sequence $x^n$ was drawn from $X$ or not. To make his decision, D relies on the knowledge of a training sequence $t_{D}^N$ drawn from $X$. On his side, A takes a sequence $y^n$ emitted by another source $Y \simeq P_Y$ still belonging to $\CC$ and tries to modify it in such a way that D thinks that the modified sequence was generated by the same source that generated $t_{D}^N$. As usual, the attacker must satisfy a distortion constraint stating that the distance between the modified sequence and $y^n$ must be lower than a threshold. Like the defender, A knows $P_X$ through a training sequence $t_{A}^K$, that in general may not coincide with $t_{D}^N$. We assume that $t_{D}^N$, $t_{A}^K$, $x^n$ and $y^n$ are generated independently. With regard to $P_Y$, we could also assume that it is known through two training sequences, one available to A and one to D, however we will see that - as for known sources and at least in the asymptotic case - such an assumption is not necessary, and hence we take the simplifying assumption that $P_Y$ is known to neither D nor A. Let, then, $H_0$ be the hypothesis that the test sequence has been generated by the same source that generated  $t_{D}^N$ and let $\Lambda$ be the acceptance region for $H_0$. In the following, we will find convenient to think of $\Lambda$ as a subset of $\XX^n \times \XX^N$, i.e., as the set of all the pairs of sequences $(x^n, t_{D}^N)$ that the defender considers to be drawn from the same source. With the above ideas in mind, and by paralleling the definition given in Section \ref{sec.KS}, we define a first version of the binary hypothesis testing game with training sequences as follows:
\begin{definition}
The $HT_{tr,a}(\SS_{D}, \SS_{A}, u)$ game is a zero-sum, strategic, game played by D and A, defined by the following strategies and payoff.
\begin{itemize}
\item{The set of strategies D can choose from is the set of acceptance regions $\Lambda$ for which the maximum false positive probability across all possible $P_X \in \CC$ is lower than a given threshold:
\begin{equation}
    \SS_{D} = \{ \Lambda : \max_{P_X \in \CC} P_X\{(x^n, t_{D}^N) \notin \Lambda\} \le P_{fp}\},
\label{eq.SFA_TR}
\end{equation}
where $P_{fp}$ is a prescribed maximum false positive probability, and where $P_X\{(x^n, t_{D}^N) \notin \Lambda\}$ indicates the probability that two independent sequences generated by $X$ do not belong to $\Lambda$. Note that the acceptance region is defined as a union of pairs of sequences, and hence $\Lambda \subset \XX^{n} \times \XX^N$.}
\item{The set of strategies A can choose from is formed by all the functions that map a sequence $y^n$ generated by $Y$ into a new sequence $z^n$ subject to a distortion constraint:
\begin{equation}
    \SS_{A} = \{ f(y^n, t_{A}^K): d(y^n, f(y^n,t_{A}^K)) \le nD\},
\label{eq.SAD_TR}
\end{equation}
where $d(\cdot,\cdot)$ is a proper distance function and $D$ is the maximum allowed per-letter distortion. Note that the function $f(\cdot)$  depends on $t_{A}^K$, since when performing his attack A can exploit the knowledge of his training sequence.}
\item{The payoff function is defined in terms of the false negative error probability, namely:
\begin{equation}
    u(\Lambda, f) = - P_{fn} =  - \hspace{-1.1cm} \sum_{\substack{t_{D}^N \in \XX^N, ~ t_{A}^K \in \XX^K \\ y^n: (f(y^n,t_{A}^K), t_{D}^N) \in \Lambda}} \hspace{-1.1cm}P_Y(y^n) P_X(t_{D}^N) P_X(t_{A}^K),
\label{eq.payoff_TR}
\end{equation}
where the error probability is averaged across all possible $y^n$ and training sequences and where we have exploited the independence of $y^n, t_{D}^N$ and $t_{A}^K$}.
\end{itemize}
\label{def.SI_ks_a}
\end{definition}

\subsection{Discussion}

Before going on with the analysis, we pose to discuss some of the choices we implicitly made with the above definition.

A first observation regards the payoff function. As a matter of fact, the expression in (\ref{eq.payoff_TR}) looks problematic, since its evaluation requires that the pmf's $P_X$ and $P_Y$ are known, however this is not the case in our scenario since we have assumed that $P_X$ is known only through $t_{D}^N$ and $t_{A}^K$, and that $P_Y$ is not known at all. As a consequence it may seem that the players of the game are not able to compute the payoff associated to a given profile and hence have no arguments upon which they can base their choice. While this is indeed a problem in a generic setup, we will show later on in the paper that asymptotically (when $n$, $N$ and $K$ tend to infinity) the optimum strategies of D and A are uniformly optimum across all $P_X$ and $P_Y$ and hence the ignorance of $P_X$ and $P_Y$ is not a problem. One may wonder why we did not define the payoff under a worst case assumption (from D's perspective) on $P_X$ and/or $P_Y$. The reason is that doing so would result in a meaningless game. In fact, given that $X$ and $Y$ are drawn from the same class of sources $\CC$, the worst case for D would always correspond to $X = Y$ for which no meaningful decision is possible\footnote{Alternatively, we could assume that $X$ and $Y$ belong to two disjoint source classes $\CC_X$ and $\CC_Y$. We leave this analysis for further research.}.

As a second remark, we stress that we decided to limit the strategies available to A to deterministic functions of $y^n$. This may seem a limiting choice, however we will see in the subsequent sections that, at least asymptotically, the optimum strategy of D depends neither on the strategy chosen by A nor on $P_Y$, then, it does not make sense for A to adopt a randomized strategy to confuse D.

A last, even more basic, comment regards the overall structure of the game. In our definition we assumed that the attacker does not intervene when $H_0$ holds, since we restricted his interest to the false negative error probability. An alternative approach could be to let the attacker modify also the sequences generated by $X$ in the attempt to increase the false positive rate. We could also depart from the Neyman-Pearson set up and define the payoff in terms of the overall error probability, or the overall Bayes risk defined on the basis of suitable cost functions associated to the two kinds of errors\footnote{In this case it would be necessary that the a-priori probabilities of the two hypotheses are known.}. While these are interesting research directions, in this paper we restrict our analysis to the game specified by Definition \ref{def.SI_ks_a}, and leave the alternative approaches for future research.

\subsection{Game variants}

Two different variants of the $HT_{tr,a}$ game are obtained by assuming a different relationship between the training sequences. In certain cases, we may assume that D has a better access to the source $X$ than A (see \cite{Gol10} for a multimedia forensics scenario in which such an assumption holds quite naturally). In our framework, we can model such a situation by assuming that the sequence $t_{A}^K$ is a subsequence of $t_{D}^N$, leading to the following definition.
\begin{definition}
The $HT_{tr,b}(\SS_{D}, \SS_{A}, u)$ game is a zero-sum, strategic, game defined as the $HT_{tr,a}$ game with the only difference that $t_{A}^K = (t_{A,l+1}, t_{A,l+2}  \dots  t_{A,l+K})$ with $l$ and $K$ known to D.
\label{def.SI_tr_b}
\end{definition}
Yet another variant is obtained by assuming that the training sequence available to A is equal to that available to D.
\begin{definition}
The $HT_{tr,c}(\SS_{D}, \SS_{A}, u)$ game is a zero-sum, strategic, game defined as the $HT_{tr,a}$ game with the only difference that $K = N$ and $t_{A}^K = t_{D}^N$ (simply indicated as $t^N$ in the following). The set of strategies of D and A are the same as in the $HT_{tr,a}$ game, while the payoff is redefined as:
\begin{equation}
    u(\Lambda, f) = - P_{fn} =  - \hspace{-1.1cm} \sum_{\substack{t^N \in \XX^N \\ y^n: (f(y^n,t^N), t^N) \in \Lambda}} \hspace{-1.1cm}P_Y(y^n) P_X(t^N).
\label{eq.payoff_TR_c}
\end{equation}

\label{def.SI_tr_c}
\end{definition}
In the rest of the paper we will first focus on version $c$ of the game, and then extend our results so to cover version $a$ as well.

\subsection{Hypothesis testing game with limited resources}

Studying the existence of an equilibrium point for the $HT_{tr,c}$ game is a prohibitive task, hence we use the same approach adopted in \cite{merhav_sabbag_08, BT13} and consider a simplified version of the game in which D can only base his decision on a limited set of statistics computed on the test and training sequences. Specifically, we require that D relies only on the relative frequencies with which the symbols in $\XX$ appear in $x^n$ and $t^N$, i.e. $P_{x^n}$ and $P_{t^N}$. Note that $P_{x^n}$ and $P_{t^N}$ are not sufficient statistics for D, since even if $Y$ is a memoryless source, the attacker could introduce some memory within the sequence as a result of the application of $f(\cdot)$. In the same way he could introduce some dependencies between the attacked sequence $z^n$ and $t^N$. It is then necessary to treat the assumption that D relies only on $P_{x^n}$ and $P_{t^N}$ as an explicit requirement.

Following \cite{merhav_sabbag_08} and \cite{BT13}, we call this version of the game {\em hypothesis testing with limited-resources}, and we refer to it as the $HT_{tr,c}^{lr}$ game. As a consequence of the limited resources assumption, $\Lambda$ can only be a union of Cartesian products of pairs of type classes, i.e. if the pair of sequences ($x^n$, $t^N$) belongs to $\Lambda$, then any pair of sequences belonging to the Cartesian product $T(P_{x^n}) \times T(P_{t^N})$ will also be contained in $\Lambda$. Since a type class is univocally defined by the empirical pmf of the sequences contained in it, we can redefine $\Lambda$ as a union of pairs of types $(P, Q)$ with $P \in \PP_n$ and $Q \in \PP_N$. In the following, we will use the two interpretations of $\Lambda$ (as a set of pairs of sequences or pairs of types) interchangeably, the exact meaning being always recoverable from the context.

We are interested in studying the asymptotic behavior of the game when $n$ and $N$ tend to infinity. To avoid the necessity of considering two limits with $n$ and $N$ tending to infinity independently, we will express $N$ as a function of $n$, and study what happens when $n$ tends to infinity. This assumption does not reduce the generality of our analysis, however it destroys the symmetry of the hypothesis testing problem with respect to the two sequences $x^n$ and $t_N$. The consequences of this loss of symmetry will be discussed in Section \ref{seq.NASHdiscuss}.

We are now ready to define the asymptotic $HT_{tr,c}^{lr}$ game. Specifically, we have:
\begin{definition}
The $HT_{tr,c}^{lr}(\SS_{D}, \SS_{A}, u)$ game is a zero-sum, strategic, game played by D and A, defined  by the following strategies and payoff:
\begin{align}
\label{eq.SD_TR_LR}
    \SS_{D} = \{ & \Lambda \subset \PP_n \times \PP_{N}: \\ \nonumber
    ~ & \max_{P_X \in \CC} P_X\{(x^n, t^{N(n)}) \notin \Lambda\} \le 2^{-\lambda n}\},
\end{align}
\begin{equation}
    \SS_{A} = \{ f(y^n, t^{N(n)}): d(y^n, f(y^n,t^{N(n)})) \le nD\},
\label{eq.SA_TR_LR}
\end{equation}
\begin{equation}
    u(\Lambda, f) = - P_{fn} =  - \hspace{-1.1cm} \sum_{\substack{t^{N(n)} \in \XX^{N(n)} \\ y^n: (f(y^n,t^{N(n)}), t^{N(n)}) \in \Lambda}} \hspace{-1.1cm}P_Y(y^n) P_X(t^{N(n)}).
\label{eq.payoff_TR_LR}
\end{equation}
\label{def.HT_tr_lr_c}
\end{definition}
Note that we ask that the false positive error probability decays exponentially fast with $n$, thus opening the way to the asymptotic solution of the game. Similar definitions can be given for versions $a$ and $b$ of the game.

\section{Asymptotic equilibrium of the $HT_{tr,c}^{lr}$ game.}
\label{sec.Equilibrium_tr_c}

We start the analysis of the asymptotic equilibrium point of the $HT_{tr,c}^{lr}$ game by determining the optimum acceptance region for D. To do so we will use an analysis similar to that carried out in \cite{Gut89} to study hypothesis testing with observed statistics. The main difference between our analysis and \cite{Gut89} is the presence of the attacker, i.e. the game-theoretic nature of our problem. The derivation of the optimum strategy for D passes through the definition of the generalized log-likelihood ratio function $h(x^n, t^N)$. Given the test and training sequences $x^n$ and $t^N$, the generalized log-likelihood ratio function is defined as (\cite{Gut89,Kendall})\footnote{To simplify the notation, when it is not strictly necessary, we omit to indicate explicitly the dependence of $N$ on $n$.}:
\begin{equation}
    h(x^n, t^{N}) = \DD(P_{x^n} || P_{r^{n+N}}) + \frac{N}{n} \DD(P_{t^N} || P_{r^{n+N}}),
\label{eq.h}
\end{equation}
where $P_{r^{n+N}}$ indicates the empirical pmf of the sequence $r^{n+N}$, obtained by concatenating $x^n$ and $t^N$, i.e.
\begin{equation}
    r_i = \left\lbrace
    \begin{array}{ll}
        x_i & i \le n \\
        t_{i-n} & n < i \le n+N
\end{array}
    \right. .
\label{eq.r}
\end{equation}
Observing that $h(x^n,t^{N})$ depends on the test and the training sequences only through their empirical pmf, we can also use the notation $h(P_{x^n} , P_{t^N})$. The study of the equilibrium for the $HT_{tr,c}^{lr}$ game passes through the following lemmas.
\begin{lemma}
For any $P_X$ we have:
\begin{align}
\label{eq.property_Dsum}
   n\DD(P_{x^n}||P_{r^{n+N}}) + & N\DD(P_{t^N} || P_{r^{n+N}}) \le \\ \nonumber
   & n\DD(P_{x^n}||P_X) + N\DD(P_{t^N} || P_X),
\end{align}
with equality holding if only if $P_X = P_{r^{n+N}}$.
\label{lemma.property_Dsum}
\end{lemma}
The proof of Lemma \ref{lemma.property_Dsum} is given in Appendix \ref{app.lemma1}.

\begin{lemma}
Let $\Lambda_{tr,c}^*$ be defined as follows:
\begin{equation}
    \Lambda_{tr,c}^* \hspace{-0.1cm} = \hspace{-0.1cm} \left\lbrace  (P_{x^n}, P_{t^N})  \hspace{-0.06cm} : h(P_{x^n}, P_{t^N}) \hspace{-0.08cm} < \hspace{-0.08cm} \lambda \hspace{-0.03cm} - \hspace{-0.03cm} |\XX| \frac{\log(n+1)(N+1)}{n}  \right\rbrace
\label{eq.optimum_SD}
\end{equation}
with
\begin{equation}
	\label{eq.Nvsn}
    	\lim_{n \rarrow \infty} \frac{log(N(n)+1)}{n} = 0.
\end{equation}
Then:
\begin{enumerate}
\item{$\max_{P_X} P_X\{(x^n, t^N) \notin \Lambda_{tr,c}^*\} \le 2^{-n(\lambda - \nu_n)}$, with $\nu_n \rarrow 0$, for $n \rarrow \infty$,}
\item{$\forall \Lambda \in \SS_{D}$, we have $\Lambda^c \subseteq \Lambda_{tr,c}^{*,c}$}.
\end{enumerate}
\label{lemma.optimum_SD}
\end{lemma}
\begin{proof}
Being $\Lambda_{tr,c}^*$ a union of pairs of types (or, equivalently, a union of Cartesian products of type classes), we have:
\begin{align}
    \max_{P_X} P_{fp} & =  \max_{P_X \in \CC}  \sum_{(x^n, t^N) \in \Lambda_{tr,c}^{*,c}} P_X(x^n,t^N) \\ \nonumber
    & = \max_{P_X \in \CC} \sum_{(P_{x^n}, P_{t^N}) \in \Lambda_{tr,c}^{*,c}} P_X(T(P_{x^n}) \times T(P_{t^N})).
\end{align}
For the class of discrete memoryless sources, the number of types with denominators $n$ and $N$ is bounded by $(n+1)^{|\XX|}$ and $(N+1)^{|\XX|}$ respectively \cite{CandT}, so we can write:
\begin{align}
     \max_{P_X} P_{fp} & \le \max_{P_X} \max_{(P_{x^n}, P_{t^N}) \in \Lambda_{tr,c}^{*,c}} \\ \nonumber
      &  \hspace{0.3cm}  [ (n+1)^{|\XX|}(N+1)^{|\XX|} P_X(T(P_{x^n}) \times T(P_{t^N})) ]  \\ \nonumber
      & \le (n+1)^{|\XX|}(N+1)^{|\XX|} \cdot \\ \nonumber
      & \hspace{0.6cm} \max_{P_X} \max_{(P_{x^n}, P_{t^N}) \in \Lambda_{tr,c}^{*,c}} 2^{-n[\DD(P_{x^n} || P_X) + \frac{N}{n}\DD(P_{t^N} || P_X) ]},
\end{align}
where in the second inequality we have exploited the independence of $x^n$ and $t^N$ and the property of types according to which for any sequence $x^n$ we have $P_X(T(P_{x^n})) \le 2^{-n \DD (P_{x^n} || P_X)}$ (see \cite{CandT}). By exploiting Lemma \ref{lemma.property_Dsum}, we can write:
\begin{align}
\label{eq.firstpartlemma}
     \max_{P_X} P_{fp} & \le (n+1)^{|\XX|}(N+1)^{|\XX|} \\ \nonumber
     & \hspace{0.4cm} \max_{(P_{x^n}, P_{t^N}) \in \Lambda_{tr,c}^{*,c}} 2^{-n[\DD(P_{x^n} || P_{r^{n+N}}) + \frac{N}{n} \DD(P_{t^N}||P_{r^{n+N}})]} \\ \nonumber
     & \le (n+1)^{|\XX|}(N+1)^{|\XX|} ~ 2^{-n(\lambda - |\XX| \frac{\log(n+1)(N+1)}{n})} \\ \nonumber
     & = 2^{-n(\lambda - 2 |\XX| \frac{\log(n+1)(N+1)}{n})},
\end{align}
where the last inequality derives from the definition of $\Lambda_{tr,c}^*$. Together with (\ref{eq.Nvsn}), equation (\ref{eq.firstpartlemma}) proves the first part of the lemma with $\nu_n = 2 |\XX| \frac{\log(n+1)(N+1)}{n}$.

For any $\Lambda \in \SS_D$, let $(x^n, t^N)$ be a generic pair of sequences contained in $\Lambda^c$, due to the limited resources assumption the cartesian product between $T(P_{x^n})$ and $T(P_{t^N})$ will be entirely contained in $\Lambda^c$. Then we have:
\begin{align}
\label{eq.secondpartlemma}
    2^{-\lambda n} & \ge \max_{P_X} P_X(\Lambda^c) \\ \nonumber
    & \stackrel{(a)}{\ge} \max_{P_X} P_X(T(P_{x^n}) \times T(P_{t^N})) \\ \nonumber
    & \stackrel{(b)}{\ge} \max_{P_X} \frac{2^{-n[\DD(P_{x^n}||P_X)+\frac{N}{n}\DD(P_{t^N}||P_X)]}}{(n+1)^{|\XX|}(N+1)^{|\XX|}} \\ \nonumber
    & \stackrel{(c)}{=} \frac{2^{-n[\DD(P_{x^n}||P_{r^{n+N}})+\frac{N}{n}\DD(P_{t^N}||P_{r^{n+N}})]}}{(n+1)^{|\XX|}(N+1)^{|\XX|}},
\end{align}
where $(a)$ is due to the limited resources assumption, $(b)$ follows from the independence of $x^n$ and $t^N$ and a lower bound on the probability of a pair of type classes \cite{CandT}, and $(c)$ derives from Lemma \ref{lemma.property_Dsum}. By taking the logarithm of both sides we find that $(x^n,t^N) \in \Lambda_{tr,c}^{*,c}$, thus completing the proof.
\end{proof}
The first part of Lemma \ref{lemma.optimum_SD} shows that, at least asymptotically, $\Lambda_{tr,c}^*$ belongs to $\SS_{D}$, while the second part implies the optimality of $\Lambda_{tr,c}^*$. An important observation is that the optimum strategy of D is univocally determined by the false positive constraint. This solves the apparent problem that we pointed out when defining the payoff of the game, namely that the payoff depends on $P_X$ and $P_Y$ and hence it is not fully known to D.  We also observe that $\Lambda_{tr,c}^*$ does not depend on $t_A^K$, hence it is the optimum defender's strategy even for versions $a$ and $b$ of the $HT_{tr}^{lr}$ game. For this reason, from now on we will simply indicate it as $\Lambda_{tr}^*$.

The most important consequence of Lemma \ref{lemma.optimum_SD} is that the optimum strategy of D does not depend on the strategy chosen by the attacker, that is $\Lambda_{tr}^*$ is a strictly dominant strategy for D. In turn this simplifies the analysis of the optimum attacking strategy. In fact, a rationale defender will surely play the dominant strategy $\Lambda_{tr}^*$, hence A can choose his strategy by assuming that D chooses $\Lambda_{tr}^*$. The derivation of the optimum attacking strategy is now an easy task. We only need to observe that the goal of A is to take a sequence $y^n$ drawn from $Y$ and modify it in such a way that:
\begin{equation}
\label{eq.attack_goal}
    h(z^n,t^N) < \lambda - |\XX| \frac{\log(n+1)(N+1)}{n},
\end{equation}
with $d(y^n, z^n) \le nD$. The optimum attacking strategy, then, can be expressed as a minimization problem, i.e.:
\begin{equation}
\label{eq.optimum_A}
    f_{tr,c}^*(y^n, t^N) = \arg \min_{z^n : d(y^n,z^n) \le nD} h(z^n,t^N).
\end{equation}
Note that to implement this strategy A needs to know $t^N$, i.e. equation (\ref{eq.optimum_A}) determines the optimum strategy only for version $c$ of the game.

Having determined the optimum strategies for D and A, we can state the first main result of the paper, summarized in the following theorem.
\begin{theorem}[Asymptotic equilibrium of the $HT_{tr,c}^{lr}$ game]
\label{theo.Nash}
The $HT_{tr,c}^{lr}$ game is a dominance solvable game and the profile $(\Lambda_{tr}^*, f_{tr,c}^*)$ is the only rationalizable equilibrium.
\end{theorem}
\begin{proof}
Lemma \ref{lemma.optimum_SD} says that $\Lambda_{tr}^*$ is a strictly dominant strategy for D, thus permitting us to eliminate all the other strategies in $\SS_D$ (since they are strictly dominated by $\Lambda_{tr}^*$). The theorem, then, follows from the optimality of $f_{tr,c}^*$ when $\Lambda_{tr}^*$ is fixed.
\end{proof}

\subsection{Discussion}
\label{seq.NASHdiscuss}

As a first remark, we observe that $(\Lambda_{tr}^*, f_{tr,c}^*)$ is the unique Nash equilibrium of the game. In addition to the properties of Nash equilibria, however, $(\Lambda_{tr}^*, f_{tr,c}^*)$ has the desirable characteristic of being the only possible choice if the two players behave rationally. In fact, a rational defender will surely adopt the acceptance region $\Lambda_{tr}^*$, since any other choice will lead to a (asymptotically) higher $P_{fn}$, regardless of the choice made by A. On his side, a rational attacker, knowing that D will behave rationally, will adopt the strategy $f_{tr,c}^*$, since this is the strategy that optimizes his payoff when D plays $\Lambda_{tr}^*$ (for more details on the notion of rationalizable equilibrium we refer to  \cite{Bern84, Pea84}).

To get a better insight into the meaning of the equilibrium point of the $HT_{tr,c}^{lr}$ game, it is instructive to compare it with the equilibrium of the corresponding game with known sources, namely the $HT_{ks}^{lr}$ game. To start with, we observe that the use of the $h$ function instead of the divergence $\DD$ derives from the fact that D must ensure that the false positive probability stays below the desired threshold for all possible sources in $\CC$. To do so, he has to estimate the pmf that better {\em explains} the evidence provided by both $x^n$ and $t^N$. This is exactly the role of $P_{r^{n+N}}$ (see equation (\ref{eq.app2})), with the generalized log-likelihood ratio corresponding to 1 over $n$ the log of the (asymptotic) probability that a source with pmf equal to $P_{r^{n+N}}$ outputs the sequences
$x^n$ and $t^N$.

Another observation regards the optimum strategy of the attacker. As a matter of fact, the functions $h(P_{x^n}, P_{t^N})$ and $\DD(P_{x^n} || P_{t^N})$ share a similar behavior: both are positive and convex functions achieving the absolute minimum when $P_{x^n} = P_{t^N}$, so one may be tempted to think that from A's point of view minimizing $\DD(P_{x^n} || P_{t^N})$ is equivalent to minimizing $h(P_{x^n}, P_{t^N})$. While this is the case in some situations, e.g. when the absolute minimum can be reached, in general the two minimization problems yield different solutions.

To further compare the $HT_{tr,c}^{lr}$ and the $HT_{ks}^{lr}$ games, it is useful to rewrite the generalized likelihood function in a more convenient way. By applying some algebra, it is easy to prove the following equivalent expression for $h$:
\begin{equation}
    h(P_{x^n}, P_{t^N}) = \DD(P_{x^n} || P_{t^N}) - \frac{N+n}{n} \DD(P_{r^{n+N}} || P_{t^N}),
\label{eq.alternative_h}
\end{equation}
showing that $h(P_{x^n}, P_{t^N})  \le \DD(P_{x^n} || P_{t^N})$ with the equality holding only in the trivial case $P_{x^n} = P_{t^N}$. This suggests that, at least for large $n$, it should be easier for A to bring a sequence generated by $Y$ within $\Lambda_{tr}^*$ than to bring it within $\Lambda_{ks}^*$. This is indeed the case, as it will be shown in Section \ref{sec.ks_vs_tr}, where we will provide a rigorous proof that the $HT_{tr,c}^{lr}$ game is actually more
favorable to the attacker than the $HT_{ks}^{lr}$ game.

We conclude this section by investigating the behavior of the optimal acceptance strategy for different values of the ratio $\frac{N}{n}$. To do so we introduce the two quantities $c_x = \frac{n}{n+N}$ and $c_t = \frac{N}{n+N}$, representing the weights of the sequences $x^n$ and $t^N$ in $r^{n+N}$. It is easy to show, in fact, that
\begin{equation}
    P_{r^{n+N}} = c_x P_{x^n} + c_t P_{t^N}.
\label{eq.weight_Pr}
\end{equation}
In the simplest case $n$ and $N$ will tend to infinity with the same speed, hence we can assume that the ratio between $N$ and $n$ is fixed, namely, $\frac{N}{n} = c \ne 0$ (we obviously have $c_x = \frac{1}{1+c}$ and $c_t = \frac{c}{1+c}$). Under this assumption, the decision of D is dictated by equation (\ref{eq.optimum_SD}) and no particular behavior can be noticed. This is not the case when $N/n$ tends to 0 or $\infty$.

If $N/n \rarrow 0$, then $P_{r^{n+N}} \rarrow P_{x^n}$ and $h(P_{x^n} , P_{t^N}) \rarrow 0$. This means that the defender will always decide in favor of $H_0$. This makes sense since when the test sequence is infinitely longer than the training sequence, the evidence provided by the training sequence is not strong enough to let the defender reject hypothesis 0.

If $N/n \rarrow \infty$, the analysis is slightly more involved. In this case $c_t \rarrow 1$ and  $P_{r^{n+N}} \rarrow P_{t^N}$, hence the first term in equation (\ref{eq.h}) tends to $\DD(P_{x^n} || P_{t^N})$. To understand the behavior of the second term of (\ref{eq.h}) when $n \rarrow \infty$, we can use the Taylor expansion of $\DD(P||Q)$ when $P$ approaches $Q$ (see \cite{CSnow}, chapter 4), which applied to the second term of the $h$ function yields:
\begin{align}
 \frac{N}{n}  \cdot \DD(P_{t^N} || P_{r^{n+N}}) & \approx  \frac{N}{2n} \cdot \sum_{x} \frac{(P_{t^N}(x) - P_{r^{n+N}}(x))^2}{P_{r^{n+N}}(x)} \nonumber \\
 & =  \frac{N}{2n} \cdot \sum_{x} \frac{(c_x P_{t^N}(x) + c_x P_{x^n}(x))^2}{P_{r^{n+N}}(x)} \nonumber \\
    & =  \frac{\frac{n}{N}}{2(\frac{n}{N} + 1)^2} \sum_{x} \frac{( P_{t^N}(x) +  P_{x^n}(x))^2}{P_{r^{n+N}}(x)}. \nonumber \\
\label{eq.TaylorApprox}
\end{align}
When $N/n \rarrow \infty$, the above expression clearly tends to 0, and hence $h(P_{x^n} , P_{t^N}) \rarrow \DD(P_{x^n} || P_{t^N})$. In other words, the optimum acceptance region tends to be equal to the one obtained for the case of know sources with $P_X$ replaced by $P_{t^N}$. This is also an intuitively reasonable result: when the training sequence is much longer than the test sequence, the empirical pmf of the training sequence provides such a reliable estimate of $P_X$ that the defender can treat it as the true pmf.

One may wonder the reason behind the asymmetric behavior of the optimum decision strategy when the length of one between the two sequences under analysis grows much faster than the other. This apparent anomaly derives from the choice of analyzing the asymptotic behavior by letting $n$ tend to infinity, a choice that breaks the symmetry between the test and training sequences. If we had defined the false positive and false negative error exponents in terms of $N$, the situation would have been completely reversed.

In the following we will always assume that $N/n = c$, since from the above analysis this turns out to be most interesting case.

\section{Analysis of the payoff at the equilibrium}
\label{sec.payoff}

Now that we have derived the equilibrium point of the $HT_{tr,c}^{lr}$ game, we are ready to analyze the payoff at the equilibrium to understand who, between the defender and the attacker is going to {\em win} the game. Our aim is to derive a result similar to Theorem \ref{theo.payoff_KS_as}, so that given two pmf's $P_X$ and $P_Y$, a false positive error exponent $\lambda$ and a distortion constraint $D$, we can derive the best achievable (for the defender) false negative error exponent $\varepsilon_{tr,c}$. Specifically, we would like to know whether it is possible for $D$ to obtain a strictly positive value of $\varepsilon_{tr,c}$, thus ensuring that the false negative error probability tends to zero exponentially fast for increasing values of $n$.

In our proofs we will find it necessary to generalize the $h$ function so that it can be applied to general pmf's not necessarily belonging to $\PP_n$ or $\PP_N$. By remembering that $N/n = c$, we introduce the following definition:
\begin{equation}
    h_c(P,Q) = \DD(P || U) + c \DD(Q || U),
\label{eq.hc}
\end{equation}
with
\begin{equation}
    U = \frac{1}{1+c} P + \frac{c}{1+c} Q.
\label{eq.hc}
\end{equation}
Note that when $P \in \PP_n$ and $Q \in \PP_N$, the above definition is equivalent to (\ref{eq.h}). By using $h_c$ instead of $h$ we can generalize the expression of the optimum acceptance region $\Lambda_{tr}^*$ so to make it possible to apply it to any pair of pmf's $P$ and $Q$ (of course when $P$ and $Q$ are not empirical pmf's the meaning of $\Lambda_{tr}^*$ as acceptance region for $H_0$ is lost):
\begin{equation}
    \Lambda_{tr}^* = \left\lbrace  (P,Q) : h_c(P,Q) < \lambda \hspace{-0.03cm} - \hspace{-0.03cm} |\XX| \frac{\log(n+1)(N+1)}{n}  \right\rbrace.
\label{eq.optimum_SD_c}
\end{equation}
With these ideas in mind, let us introduce the set $\Gamma^n_{tr,c}$ containing all the pairs of sequences  $(y^n, t^N)$, for which A is able to bring $y^n$ within $\Lambda_{tr}^{*,n}$ (for sake of clarity we use the apex $n$ to explicitly indicate that $\Lambda_{tr}^{*,n}$ refers to pairs of sequences respectively of length $n$ and $N = cn$):
\begin{align}
\label{eq.Gamman_TR}
\Gamma^n_{tr,c} =  \{ & (y^n, t^N) : \exists z^n \text{ s.t. } \nonumber \\
& (z^n, t^N) \in \Lambda_{tr}^{*,n} \text{ and } d(y^n, z^n) \le nD \}.
\end{align}

By observing that $\Gamma^n_{tr,c}$ depends on $t^N$ only through $P_{t^N}$ and by reasoning as in the proof of Property 1 in \cite{BT13} (we need to assume that the distance measure $d$ is permutation-invariant), we can show that $\Gamma^n_{tr,c}$ is still a union of pairs of type classes, and hence we can redefine it as:
\begin{align}
\Gamma^n_{tr,c}  = \{ & (P_{y^n}, P_{t^N}) : \forall y^n \in T(P_{y^n}) \quad \exists z^n \text{ s.t. } \nonumber \\
 &  (P_{z^n}, P_{t^N}) \in \Lambda_{tr}^{*,n} \text{ and } d(y^n, z^n) \le nD\}.
\label{eq.Gamman_TR_types}
\end{align}
Note that, by adopting the generalized version of $\Lambda_{tr}^*$ in which $h_c$ is used instead of $h$, the above definition can also be applied when $P_{t^N} $ is replaced by a generic pmf $Q$ not necessarily belonging to $\PP_N$. We will also find it convenient to fix $Q$ and consider the set of types $P_{x^n}$ for which $(P_{x^n}, Q)$ belongs to $\Lambda^{*,n}_{tr}$ and $\Gamma^n_{tr,c}$, that is:
\begin{equation}
\Lambda^{*,n}_{tr} (Q)  = \{P_{x^n} : (P_{x^n}, Q) \in \Lambda^{*,n}_{tr} \},
\label{eq.Lambdan_TR_types_1dim}
\end{equation}
\begin{align}
\Gamma^n_{tr,c} (Q)  = \{ & P_{y^n} : \forall y^n \in T(P_{y^n})  \quad \exists z^n \text{ s.t. }  \nonumber \\
 &P_{z^n} \in \Lambda_{tr}^{*,n}(Q) \text{ and } d(y^n, z^n) \le nD\}.
\label{eq.Gamman_TR_types_1dim}
\end{align}
The derivation of the false negative error exponent at the equilibrium passes through the following asymptotic extension of $\Gamma^n_{tr,c} (Q)$:
\begin{equation}
    \Gamma_{tr,c}^{\infty} (Q) = cl \left( \bigcup_n \Gamma^n_{tr,c} (Q)  \right).
\label{eq.Gamman_TR_types_1dim_inf}
\end{equation}
The importance of the above definition is that for any source $P_X$, given the false positive error exponent $\lambda$ and the maximum allowed per-letter distortion $D$, the set $\Gamma_{tr,c}^{\infty} (P_X)$ corresponds to the indistinguishability region of the $HT_{tr,c}^{lr}$ game, i.e. the set of all the pmf's for which D does not succeed in distinguishing between $H_0$ and $H_1$ ensuring a false negative error probability that tends to zero exponentially fast. In other words, if $P_Y \in \Gamma_{tr,c}^{\infty} (P_X)$, no strictly positive false negative error exponent can be achieved by D. To prove that this is indeed the case, we need to prove the following theorem.
\begin{theorem}[Asymptotic payoff of the $HT_{tr,c}^{lr}$ game]
For the $HT_{tr,c}^{lr}$ game, with $N/n = c$ and assuming an additive distortion measure, the false negative error exponent at the equilibrium is given by
\begin{equation}
    \varepsilon_{tr,c} = \min_{Q} [ c \cdot \DD(Q || P_X) + \min_{P \in \Gamma_{tr,c}^{\infty}(Q) } \DD (P || P_Y)].
\label{eq.fnerr_exp_TRc}
\end{equation}
\label{theo.SanovTRc}
\end{theorem}
\begin{proof}
By using the definitions given in this section, the false negative error probability at the equilibrium, for  a given $n$, can be written:
\begin{eqnarray}
    P_{fn} & = & \sum_{t^N} P_{X}(t^N) \hspace{-0.3cm}\sum_{y^{n}\in
    \Gamma_{tr,c}^n(P_{t^N})} P_{Y} (y^{n})\nonumber\\
    & = & \sum_{Q \in \PP_N} P_{X}(T(Q)) \sum_{P \in
    \Gamma_{tr,c}^n(Q)} P_{Y} (T(P)).
\label{eq.Theo4_Pfn}
\end{eqnarray}
We start by deriving an upper-bound of the false negative error probability. By exploiting some well-known bounds on the probability of a type class and the number of types in $\PP_n$ \cite{CandT}, we can write:
\begin{eqnarray}
   P_{fn} & \leq  & \sum_{Q \in \PP_N} P_X(T(Q)) \sum_{P \in \Gamma_{tr,c}^n
   (Q)} 2^{- n \DD(P || P_Y)} \nonumber \\
   & \leq &  \sum_{Q \in \PP_N} P_X (T(Q)) (n +  1)^{|\mathcal X|} 2^{- n
   \min\limits_{P \in \Gamma_{tr,c}^n (Q)} \DD(P || P_Y)}
   \nonumber \\
   & \le & \sum_{Q \in \PP_N} P_X (T(Q)) (n + 1)^{|\mathcal X|} 2^{-
   n \min\limits_{P \in \Gamma_{tr,c}^{\infty}(Q)} \DD(P || P_Y)}\nonumber\\
   & \leq & (n + 1)^{|\mathcal X|} (N + 1)^{|\mathcal X|} \nonumber \\
   & & \cdot 2^{- n \min\limits_{
   Q \in \PP_N } [\frac{N}{n} \DD(Q || P_X) + \min\limits_{P \in \Gamma_{tr,c}^{\infty}(Q)} \DD( P ||
   P_Y)]}\nonumber\\
   & \leq & (n + 1)^{|\mathcal X|} (N + 1)^{|\mathcal X|} \nonumber \\
   & & \cdot 2^{- n \min\limits_{
   Q \in \CC} [c \DD(Q || P_X) + \min\limits_{P \in \Gamma_{tr,c}^{\infty}(Q)} \DD( P ||
   P_Y)]},
   \label{eq.low_bound_P_fn1}
\end{eqnarray}
where the last inequality is obtained by minimizing over all $Q$ without requiring that $Q \in \PP_N$. By taking the log and dividing by $n$ we find:
\begin{eqnarray}
- \frac{\log P_{fn}}{n}  \ge \min\limits_{Q \in \CC} \big[c D( Q || P_X) + \min\limits_{P
\in \Gamma_{tr,c}^{\infty} (Q)} D( P || P_Y)\big] + \alpha_n,
 \label{eq.low_bound_P_fn2}
\end{eqnarray}
with $\alpha_n = |\XX| \frac{\log(n+1)(N+1)}{n}$ tending to 0 when $n$ tends to infinity.

We now turn to the analysis of a lower bound for $P_{fn}$. Let  $Q^*$ be the pmf achieving the minimum in (\ref{eq.fnerr_exp_TRc}). Due to the density of rational numbers within real numbers, we can find a sequence of pmf's $Q_n \in \PP_n$ that tends to $Q^*$ when $n$ tends to infinity. By remembering that $N = nc$, the subsequence $Q_N = Q_{nc}$ will also tend to $Q^*$ when $n$ (and hence $N$) tends to infinity\footnote{In order to simplify the analysis, we assume that $c$ is a non-null integer value, the extension of the proof to non-integer values of $c$ is tedious but straightforward.}. Let us now consider the following sequence of inequalities:
\begin{align}
\label{eq.up_bound_P_fn}
    P_{fn} & \stackrel{(a)}\ge \sum_{Q \in \PP_N} P_X(T(Q)) \sum_{P \in \Gamma_{tr,c}^n (Q)} \frac{2^{-n \DD(P || P_Y)}}{(n+1)^{|\XX|}} \\ \nonumber
    & {\ge} \sum_{Q \in \PP_N} P_X(T(Q)) \frac{2^{-n \min\limits_{P \in \Gamma_{tr,c}^n(Q)} \DD(P || P_Y) }}{(n+1)^{| \XX |}} \\ \nonumber
    & \stackrel{(b)}{\ge} \sum_{Q \in \PP_N} \frac{2^{-N\DD(Q || P_X)}}{(N+1)^{|\XX|}} \frac{2^{-n \min\limits_{P \in \Gamma_{tr,c}^n(Q)} \DD(P || P_Y) }}{(n+1)^{| \XX |}} \\ \nonumber
    & = \sum_{Q \in \PP_N} \frac{2^{-n [c \DD(Q || P_X) + \min\limits_{P \in \Gamma_{tr,c}^n(Q)} \DD(P || P_Y)]}}{(N+1)^{|\XX|}(n+1)^{|\XX|}} \\ \nonumber
    & \stackrel{(c)}\ge \frac{2^{-n [c \DD(Q_N || P_X) + \min\limits_{P \in \Gamma_{tr,c}^n(Q_N)} \DD(P || P_Y)]}}{(N+1)^{|\XX|}(n+1)^{|\XX|}},
\end{align}
where inequalities (a) and (b) derive from a known lower bounds on the probability of a type class \cite{CandT}, and in (c) we have replaced the sum with a single element of the subsequence $Q_N$ defined previously. By taking the log and dividing by $n$, we obtain
\begin{equation}
    -\frac{\log P_{fn}}{n} \le c \DD(Q_N || P_X) + \min\limits_{P \in \Gamma_{tr,c}^n(Q_N)} \DD(P || P_Y) + \beta_n,
\label{eq.up_bound_P_fn_2}
\end{equation}
where $\beta_n = |\XX| \frac{\log(n+1)(N+1)}{n}$ tends to 0 when $n$ tends to infinity.
To continue, let $P^*$ be defined as follows
\begin{equation}
    P^* = \arg\min\limits_{P \in \Gamma_{tr,c}^{\infty}(Q^*)} \DD(P || P_Y).
\label{eq.Pstar}
\end{equation}
In Appendix \ref{app.lemmaTondi}, we show that it is possible to find a sequence $P_n$, where each $P_n$ belongs to ${\Gamma_{tr,c}^n} (Q_N)$, that tends to $P^*$ when $n$ tends to infinity. By starting from equation (\ref{eq.up_bound_P_fn_2}) and by exploiting the continuity of the divergence function, for $n$ large enough we can write
\begin{align}
\label{eq.up_bound_P_fn_3}
    -\frac{\log P_{fn}}{n} & \le c \DD(Q^* || P_X) + \beta'_n + \DD(P_n || P_Y) + \beta_n, \\ \nonumber
    & \le c \DD(Q^* || P_X) + \beta'_n + \DD(P^* || P_Y) + \beta''_n + \beta_n,
\end{align}
where all the sequences $\beta_n$, $\beta'_n$ and $\beta''_n$ tend to zero when $n$ tends to infinity.

By coupling equations  (\ref{eq.low_bound_P_fn2}) and (\ref{eq.up_bound_P_fn_3}) and by letting $n \rarrow \infty$, we eventually obtain:
\begin{equation}
    -\lim_{n \rarrow \infty} \frac{\log P_{fn}}{n} = \min_{Q} [ c \cdot \DD(Q || P_X) + \min_{P \in \Gamma_{tr,c}^{\infty}(Q) } \DD (P || P_Y)],
\label{eq.fn_err_exp}
\end{equation}
thus proving the theorem.
\end{proof}

According to Theorem \ref{theo.SanovTRc}, we can distinguish two cases depending on the relationship between $P_X$ and $P_Y$.
 \begin{enumerate}\label{e.e.cases_tr}
    \item $P_Y \in \Gamma_{tr,c}^{\infty} (P_X)$ \quad then \quad $\varepsilon_{tr,c} = 0$;
    \item $P_Y \notin \Gamma_{tr,c}^{\infty} (P_X)$ \quad then \quad $\varepsilon_{tr,c} > 0$.
\end{enumerate}
In the former case, which is obtained by letting $Q^* = P_X$, it is not possible for D to obtain a strictly positive false negative error exponent while ensuring that the false positive error exponent is at least equal to $\lambda$. In the latter case, it is not possible that the two divergences in (\ref{eq.fnerr_exp_TRc}) are simultaneously equal to zero, hence $P_{fn}$ tends to 0 exponentially fast. In other words, given $\lambda$ and $D$, the condition $P_Y \notin \Gamma_{tr,c}^{\infty}(P_X)$ ensures that the {\em distance} between $P_Y$ and $P_X$ is large enough to allow a reliable distinction between sequences drawn from $P_X$ and sequences drawn from $P_Y$ despite the presence of the adversary. As anticipated, then, $\Gamma_{tr,c}^{\infty}(P_X)$ is the indistinguishability region of the $HT_{tr,c}^{lr}$ game.

\subsection{Comparison between the $HT_{ks}^{lr}$ and $HT_{tr,c}^{lr}$ games}
\label{sec.ks_vs_tr}

In this section we compare the asymptotic performance achievable by D for the $HT_{ks}^{lr}$ and $HT_{tr,c}^{lr}$ games. We start the analysis by comparing the indistinguishability regions of the two games, namely $\Gamma_{ks}^{\infty}(P_X)$ and $\Gamma_{tr,c}^{\infty}(P_X)$ (where, as opposed to Section \ref{sec.KS}, we now explicitly indicate the dependence of $\Gamma_{ks}^{\infty}$ on $P_X$).

The comparison between the two regions relies on the comparison between the divergence and the generalized likelihood function. In particular, the starting point of our analysis is the following lemma.

\begin{lemma}[Relationship between $h_c$ and $\DD$]
\label{lemma.alternative_vc}
Let $N/n = c$, with $c \ne 0$, $c \ne \infty$, for any $P \ne P_X$ we have,
\begin{equation}
h_c(P,P_X) < \DD(P||P_X).
\label{eq.h_vs_D}
\end{equation}
\end{lemma}
\begin{proof}
By rewriting $h_c(P, P_X)$ as in equation (\ref{eq.alternative_h}), we have:
\begin{equation}
    h_c(P, P_X) = \DD(P||P_X) - (1+c) \DD(U||P_X)
\label{eq.alternative_vc}
\end{equation}
with $U = P/(1+c) + cP_X/(1+c)$, which is equal to $P_X$ if and only if $P = P_X$, when we have $\DD(U||P_X) = 0$ thus yielding $h_c(P, P_X) = \DD(P||P_X) = 0$.
\end{proof}
In the subsequent proofs we will refer to the way the mapping function $f$ operates on $y^n$ to produce $z^n$\footnote{To keep the notation as light as possible we will not distinguish between mapping functions used for the known source case and those applying to hypothesis testing with training sequences, even if, rigorously speaking, these are quite different functions since the latter also depend on the training sequence $t^N$.}. Specifically, we will indicate with $n_f(i \rarrow j)$ the number of times that $f$ transforms the $i$-th symbol of $\XX$ into the $j$-th one.
The main result of our analysis is stated in the following theorem.
\begin{theorem}[$HT_{tr,c}^{lr}$ vs $HT_{ks}^{lr}$]
For any finite, non-null value of $c$, any $P_X$, $\lambda > 0$ and $D$ we have
\begin{equation}
    \Gamma_{ks}^{\infty}(P_X) \subset \Gamma_{tr,c}^{\infty}(P_X).
\label{eq.inclusion}
\end{equation}
\label{theo.inclusion}
\end{theorem}
\begin{proof}
We will prove the theorem by first showing that $\Gamma_{ks}^{\infty}(P_X) \subseteq \Gamma_{tr,c}^{\infty}(P_X)$, and then finding at least one point (actually an infinite set of points) that belongs to $\Gamma_{tr,c}^{\infty}(P_X)$ but does not stay in $\Gamma_{ks}^{\infty}(P_X)$.

Let $y^n$ be a sequence such that $P_{y^n} \in \Gamma_{ks}^n(P_X)$, this means that a mapping $f^n$ exists that transforms $y^n$ into a sequence $z^n$ belonging to $\Lambda_{ks}^{*,n}(P_X)$, while satisfying the distortion constraint.
Let now $m$ be a multiple of $n$ ($m = kn$). For $k$ large enough we have
\begin{equation}
\label{relation threshold}
    \lambda - |\XX| \frac{\log(n + 1)}{n} < \lambda - |\XX| \frac{\log(m + 1)(cm + 1)}{m}.
\end{equation}
Since $P_{z^n} \in \Lambda_{ks}^{*,n}(P_X)$, Lemma \ref{lemma.alternative_vc} permits us to write:
\begin{eqnarray}
\label{eq.Lambda_weakinclusion}
    h_c(P_{z^n},P_X) & \le & \DD(P_{z^n} || P_X) \\ \nonumber
    & < & \lambda - |\XX| \frac{\log (n+1)}{n} \\ \nonumber
    & < & \lambda - |\XX| \frac{\log (m+1)(cm+1)}{m}.
\end{eqnarray}

Given that $m$ is a multiple of $n$, any $P \in \PP_n$ also belongs to $\PP_m$, permitting us to conclude that $P_{z^n} \in \Lambda_{tr}^{*,m}(P_X)$. Let now $y^m$ be an $m$-long sequence having the same type of $y^n$ (this is surely possible since $m$ is a multiple of $n$). If we apply to $y^m$ a mapping function $v^m$ for which $n_{v^m}(i \rarrow j) = k n_{f^n}(i \rarrow j)$, the sequence $z^m = v^m(y^m)$ will have the same type of $z^n$, and hence by virtue of equation (\ref{eq.Lambda_weakinclusion}) $P_{z^m} \in \Lambda_{tr}^{*,m}(P_X)$. In addition, the mapping $v^m$ introduces the same per-letter distortion of $f^n$ for any additive distortion measure, permitting us to conclude that $P_{y^m} \in \Gamma_{tr,c}^m(P_X)$. In summary, we have shown that for any $P \in \Gamma_{ks}^n(P_X)$ an $m = kn$ exists such that $P \in \Gamma_{tr,c}^m(P_X)$, and hence:
\begin{eqnarray}
\label{eq.weak_inclusion_Gamma}
    \Gamma_{ks}^{\infty}(P_X) & = & cl \left( \bigcup_n \Gamma^n_{ks} (P_X)  \right) \\ \nonumber
    & \subseteq & cl \left( \bigcup_m \Gamma^m_{tr,c} (P_X)  \right) =  \Gamma_{tr,c}^{\infty}(P_X).
\end{eqnarray}

We now prove that there is at least one point that belongs to $\Gamma_{tr,c}^{\infty}(P_X)$ but does not belong to $\Gamma_{ks}^{\infty}(P_X)$ (actually there is an infinite number of such points). To do so let us consider a point $P^*$ belonging to the boundary of $\Gamma_{tr,c}^{\infty}(P_X)$. Since by definition $\Gamma_{tr,c}^{\infty}(P_X)$ is a closed set, $P^*$ will also belong to it. Due to Lemma \ref{Lemma_density_Gamma_n_in_var} (Appendix \ref{app.lemmaTondi}), we can find a sequence of types $P_n^i \in \Gamma_{tr,c}^n(P_X)$ that tends to $P^*$ from the inside of $\Gamma_{tr,c}^{\infty}(P_X)$. On the other hand, since $P^*$ lies on the boundary of the closed set $\Gamma_{tr,c}^{\infty}(P_X)$, any ball centered in $P^*$ will contain an infinite number of points that do not belong to $\Gamma_{tr,c}^{\infty}(P_X)$. Due to the density of rational numbers in real numbers, it is possible to define an outer sequence of types $P_n^o$ tending to $P^*$, for which $P_n^o \in \PP_n$ and for which no mapping function exists that when applied to the sequences in the type class of $P_n^o$ moves them into $\Lambda_{tr}^{*,n}(P_X)$ with a per-letter distortion equal or lower than $D$. In other words, for any sequence $y^n \in P_n^o$ and any (distortion-limited) mapping $f(y^n) = w^n$ we have
\begin{equation}
    h_c(P_{w^n},P_X) \ge \lambda - |\XX| \frac{\log(n+1)(cn+1)}{n}.
\label{eq.out_of_gamma}
\end{equation}
Let $f^{n}$ be the sequence of optimum mapping functions that applied to the sequences in $T(P_n^i)$ results in a sequence $z^n$ for which $h_c(P_{z^n},P_X) < \lambda - |\XX| [\log(n+1)(cn+1)]/n$.

Given that $P_{z^n}$ is obtained by transforming sequences belonging to a sequence of types tending to the limit type $P^*$, by continuity we can say that the sequence of types $P_{z^n}$ will also converge to a type, say $P_z^*$. Let us assume now that $P_z^* \ne P_X$\footnote{It is always possible to find a point $P^*$ for which this is true, unless the optimal acceptance region asymptotically reduces to the single point $P_X$. This would be the case if we let $\lambda \rarrow 0$, a situation that is not considered in the present analysis.}. Of course we also have $P_{z^n} \ne P_X$ (at least for large $n$). Then, by Lemma \ref{lemma.alternative_vc}, and due to the continuity of the $h_c$ function, we have
\begin{align}
\label{eq.tau}
    & \DD(P_{z^n}||P_X) - h_c(P_{z^n},P_X)  = \tau_n, \nonumber \\
    & \DD(P^*_{z}||P_X) - h_c(P^*_{z},P_X)  = \tau, \nonumber \\
    & \lim_{n \rarrow \infty} \tau_n  = \tau,
\end{align}
where $\tau$ is strictly larger than 0.

Let us now consider the outer sequence $P_n^o$. Since both $P_n^i$ and $P_n^o$ tend to $P^*$, for $n$ large enough, $P_n^i$ and $P_n^o$ will become arbitrarily close. By continuity, if we apply the same\footnote{Rigorously speaking it is possible that $f^n$ can not be applied {\em as is} to the sequences in $T(P_n^o)$, since we have to ensure that, for each $i$, $n_{f^n}(i \rarrow j)$ is not larger than the number of symbols $i$ in the to-be-mapped sequence. However, given the closeness of the sequences in $P_n^i$ and $P_n^o$ the modifications we need to introduce within $f^n$ are minor and our arguments still work. Readers may refer to Appendix A in \cite{BT13} for a detailed proof of how the closeness of types can be exploited to ensure that the types of the remapped sequences are also close to each other.} $f^{n}$ to a sequence in $T(P_n^o)$, we will obtain a sequence $w^n$ whose type is arbitrarily close to $P_{z^n}$, let us indicate it by $P_{w^n}$ (we can also say that $P_{w^n} \rarrow P_z^*$).
Due to the continuity of the $h_c$ function, $h_c(P_{z^n},P_X)$ and $h_c(P_{w^n},P_X)$ will also be arbitrarily close; $h_c(P_{w^n},P_X)$, though, will be larger than or equal to $\lambda - |\XX| (\log(n+1)(cn+1))/n$, since $w^n$ has been obtained by starting from a point belonging to $P_n^o$. This, in turn, means that $h_c(P_{z^n},P_X)$ will be arbitrarily close to $\lambda - |\XX| (\log(n+1)(cn+1))/n$ (though lower than that). In other words, for any $\delta >0$, when $n$ is large enough, we have
\begin{equation}
    \left| \left( \lambda - |\XX| \frac{\log(n+1)(cn+1)}{n}\right) - h_c(P_{z^n},P_X) \right| \le \delta.
\label{eq.closeness}
\end{equation}
If $n$ is large enough, then, for all the sequences in $T(P_n^i)$ we have:
\begin{eqnarray}
\label{eq.strong_inclusion_prequel}
    \DD(P_{z^n}||P_X) & = & h_c(P_{z^n}, P_X) + \tau_n \\ \nonumber
    & \ge & \lambda - |\XX| \frac{\log(n+1)(cn+1)}{n} - \delta + \tau_n \\ \nonumber
    & > & \lambda,
\end{eqnarray}
where for the last inequality we have exploited the fact that for large $n$, $\delta$ can be made arbitrarily small, $\log(n)/n$ tends to 0 and $\tau_n$ tends to $\tau > 0$.
As a result, for large $n$, $P_n^i \notin \Gamma_{ks}^{m}(P_X)$, for any $m$ and hence:
\begin{equation}
\label{eq.strong_inclusion}
   P_n^i \notin \bigcup_n \Gamma_{ks}^n(P_X).
\end{equation}
On the other hand, $P_n^i$ can not belong to the closure of $\cup_n \Gamma_{ks}^n(P_X)$ (and hence to $\Gamma_{ks}^{\infty}(P_X)$), since in this case we could find a new sequence of types arbitrarily close to $P_n^i$ which could be moved within $\Lambda_{ks}^{*,n}$. By continuity, then, $P_n^i$ could also be moved within $\Lambda_{ks}^{*,n}$ for some $n$ thus contradicting (\ref{eq.strong_inclusion_prequel}). It goes without saying that, a fortiori, no point $P^*$ on the boundary of $\Gamma_{tr,c}^{\infty}(P_X)$ belongs to $\Gamma_{ks}^{\infty}(P_X)$.
\end{proof}
Theorem \ref{theo.inclusion} has two simple corollaries.
\begin{corollary}
\label{Corollary1_strict_inclusion}
For any pmf $P$ belonging to the boundary of $\Gamma_{ks}^{\infty}(P_X)$ there exists a positive value $\varepsilon$ such that $B(P,\varepsilon) \subset \Gamma_{tr,c}^{\infty}(P_X)$, where $B(P,\varepsilon)$ is a ball centered in $P$ with radius $\varepsilon$. In the same way, for any pmf $P$ belonging to the boundary of $\Gamma_{tr,c}^{\infty}(P_X)$ there exists a positive value $\varepsilon$ such that $B(P,\varepsilon) \cap \Gamma_{ks}^{\infty}(P_X) = \emptyset$.
\end{corollary}
\begin{proof}
It follows immediately from the proof of Theorem \ref{theo.inclusion}.
\end{proof}
\begin{corollary}
\label{Corollary2_strict_inclusion}
Let $\varepsilon_{ks}$ and $\varepsilon_{tr,c}$ denote the error exponents at the equilibrium for the $HT_{ks}^{lr}$ and $HT_{tr,c}^{lr}$ games. Then we have:
\begin{equation}
    \varepsilon_{tr,c} \le \varepsilon_{ks},
\end{equation}
where the equality holds if and only if $P_Y \in  \Gamma_{ks}^{\infty}(P_X)$, when both error exponents are equal to 0.
\end{corollary}

\begin{proof}
The corollary is obvious when $P_Y \in \Gamma_{tr,c}^{\infty}(P_X)$, since in this case $\varepsilon_{tr,c} = 0$ and $\varepsilon_{ks} = 0$ if $P_Y \in \Gamma_{ks}^{\infty}(P_X)$ and nonzero otherwise. If $P_Y \notin \Gamma_{tr,c}^{\infty}(P_X)$, by considering the expression of the error exponent for the $HT_{tr,c}^{lr}$ game we have:
\begin{eqnarray}
\label{epsilon_comparison_1}
\varepsilon_{tr,c} & = & \min_{Q \in \mathcal P} [ c \cdot \DD(Q || P_X) + \min_{P \in \Gamma_{tr,c}^{\infty} (Q)} \DD(P || P_Y)] \\ \nonumber
 & \le & c \DD(P_{X} || P_X) + \min_{P \in \Gamma_{tr,c}^{\infty}(P_X)} \DD(P||P_{Y}) \\ \nonumber
 &\stackrel{(a)} =  & \min_{P \in \Gamma_{tr,c}^{\infty}(P_X)} \DD(P||P_{Y}) \\ \nonumber
 & < & \min_{P \in \Gamma_{ks}^{\infty}(P_X)} \DD(P||P_{Y}) = \varepsilon_{ks}.
\end{eqnarray}
where the last strict inequality is justified by observing that the absolute minimum of $\DD(P || P_Y)$ is obtained for $P = P_Y$ which lies outside $\Gamma_{tr,c}^{\infty}(P_X)$, hence due to the convexity of $\DD$ and Corollary \ref{Corollary1_strict_inclusion}, the value $P$ satisfying the minimization on the left-hand side of equality ($a$) belongs to the non-empty set $\{\Gamma_{tr,c}^{\infty}(P_X) / \Gamma_{ks}^{\infty}(P_X)\}$.
\end{proof}

Theorem \ref{theo.inclusion} and Corollary \ref{Corollary2_strict_inclusion} permit us so to conclude that hypothesis testing with training data is more favorable to the attacker than hypothesis testing with known sources. The reason behind such a result is the use of the $h$ function instead of the divergence, which in turns stems from the need for the defender to ensure that the constraint on the false positive error probability is satisfied for all $P_X \in \CC$. It is such a worse case assumption that ultimately favors the attacker in the $HT_{tr,c}^{lr}$ game.

\section{Binary hypothesis testing game with independent training sequences ($HT_{tr,a}^{lr}$).}
\label{sec.version_a}

We now pass to the analysis of version $a$ of the $HT_{tr}^{lr}$ game. We remind that in this case D and A rely on independent training sequences, namely $t_D^N$ and $t_A^K$. As for version $c$, we assume that both $N$ and $K$ grow linearly with $n$ and that the asymptotic analysis is carried out by letting $n$ go to infinity. Specifically, we assume that $N = cn$ and $K = dn$. As we already noted in Section \ref{sec.Equilibrium_tr_c}, the strategy $\Lambda_{tr}^*$ identified by Lemma \ref{lemma.optimum_SD} is optimum regardless of the relationship between $t_D^N$ and $t_A^K$, hence the only difference between versions $a$ and $c$ of the game is in the strategy of the attacker. In fact, now the attacker does not have a perfect knowledge of the acceptance region adopted by the defender, since such a region depends on the empirical pmf of $t_D^N$ which A does not know.

A reasonable strategy for the attacker could be to use the empirical pmf of $t_A^K$ instead of the one derived from $t_A^N$. More precisely, by using the notation introduced in Section \ref{sec.payoff} (equation (\ref{eq.Lambdan_TR_types_1dim})), the attacker could try to move $y^n$ into $\Lambda_{tr}^{*,n}(P_{t_A^K})$, while the acceptance region adopted by the defender is $\Lambda_{tr}^{*,n}(P_{t_D^N})$. Given that $t_D^N$ and $t_A^K$ are generated by the same source, their empirical pmf's will both tend to $P_X$ when $n$ goes to infinity, and hence using $\Lambda_{tr}^{*,n}(P_{t_A^K})$ should be {\em in some way}  equivalent to using $\Lambda_{tr}^{*,n}(P_{t_D^N})$.
In fact, in the following we will show that, given $P_X$, $D$ and $\lambda$, the indistinguishability region for version $a$ of the game, let us call it $\Gamma_{tr,a}^{\infty}(P_X)$, is identical to the indistinguishability region of version $c$. Of course, this does not mean that the achievable payoff for the $HT_{tr,a}^{lr}$ game is equal to that of the $HT_{tr,c}^{lr}$ game, since, even if the indistinguishability region is the same, outside it the false negative error exponent for case $a$ may be different (actually larger) than that of case $c$.

We start our analysis by assuming that $c = d$ (and hence $N = K$), i.e. the training sequences available to the defender and the attacker have the same length.
Our goal is to investigate the asymptotic behavior of the payoff of the $HT_{tr,a}^{lr}$ game for the profile $(\Lambda_{tr}^{*,n}(P_{t_D^N}), \tilde{f})$, where the, not necessarily optimum, strategy $\tilde{f}$ adopted by the attacker is defined as:
\begin{equation}
    \tilde{f}(y^n,t_A^N) = \arg \min_{z^n : d(z^n, y^n) \le nD} h(P_{z^n}, P_{t_A^N}).
\label{eq.subopt}
\end{equation}
We also impose the additional constraint that for any permutation $\sigma$ we have:
\begin{equation}
    \tilde{f}(\sigma(y^n),t_A^N) = \sigma(\tilde{f}(y^n,t_A^N)).
\label{eq.permut}
\end{equation}
By following the same flow of ideas used in Section \ref{sec.payoff}, we consider the set of sequences for which the attacker is able to move $y^n$ within the acceptance region $\Lambda_{tr}^{*,n}(P_{t_D^N})$, i.e.:
\begin{equation}
    \tilde{\Gamma}_{tr,a}^n = \{(y^n, t_{D}^N,t_{A}^N) : \tilde{f}(y^n,t_{A}^N) \in \Lambda_{tr}^{*,n}(P_{t_D^N}) \}.
\label{eq.gamma_triple}
\end{equation}
Thanks to the additional constraint in equation (\ref{eq.permut}), and by reasoning as in the proof of Property 1 in \cite{BT13}, it is easy to show that $\tilde{\Gamma}_{tr,a}^n $ is a union of triple of type classes\footnote{The necessity of imposing that $\tilde{f}$ commutes with permutations comes out when for a certain $y^n$ the minimization in (\ref{eq.subopt}) has several solutions $\{z^n(1) \dots z^n(k)\}$, some of which, say the first $m$, fall within $\Lambda_{tr}^{*,n}(P_{t_D^N})$ while the others don't. If we consider $\sigma({y^n})$ instead of $y^n$, $\{\sigma(z^n(1)) \dots \sigma(z^n(k))\}$ will still be solutions of the minimization problem. In addition, the first $m$ sequences will still belong to $\Lambda_{tr}^{*,n}(P_{t_D^N})$, while the others will not. If we want that $\tilde{\Gamma}_{tr,a}^n$ is a union of triple of type classes, it is necessary to require that when $y^n$ is permuted $\tilde{f}$ continues to pick up a minimizer inside (or outside) $\Lambda_{tr}^{*,n}(P_{t_D^N})$. This is surely the case if $\tilde{f}(\sigma(y^n),t_A^N) = \sigma(\tilde{f}(y^n,t_A^N))$.}, hence permitting us to redefine $\tilde{\Gamma}_{tr,a}^n$ in terms of types. Similarly to version $c$ of the game, we find it useful to introduce the following definition:
\begin{align}
\label{eq.gamma_f2}
    \tilde{\Gamma}_{tr,a}^n (P_{t_D^N}, P_{t_A^N}) = & \{P_{y^n} \in \PP_n : \\ \nonumber
    & \forall y^n \in T(P_{y^n}),  ~ (\tilde{f}(y^n, t_{A}^N), t_D^N) \in \Lambda_{tr}^{*,n}\}.
\end{align}
By using the generalized function $h_c$ instead of $h$, we can apply the above definition to any pair of pmf's. Specifically, given two pmf's $Q$ and $R$, we define:
\begin{align}
\label{eq.gamma_f2_bis}
    \tilde{\Gamma}_{tr,a}^n (Q,R) =  \{ & P_{y^n} \in \PP_n : \\ \nonumber
    & \forall y^n \in T(P_{y^n}),  (\tilde{f}(y^n,R),Q) \in \Lambda_{tr}^{*,n}.
\end{align}
It is easy to see that:
\begin{align}
\label{eq.QeqR}
    \tilde{\Gamma}_{tr,a}^n (Q,R) & \subseteq \tilde{\Gamma}_{tr,a}^n (Q,Q) \\ \nonumber
     \tilde{\Gamma}_{tr,a}^n (Q,Q) & = \Gamma_{tr,c}^n (Q),
\end{align}
since when (and only when) $Q = R$ A performs its attack by using exactly the same acceptance region adopted by D, while in all the other cases he can rely only on an estimate based on its own training sequence. Paralleling the analysis of the $HT_{tr,c}^{lr}$ game, we introduce the set
\begin{equation}
    \tilde{\Gamma}_{tr,a}^{\infty} (Q,R) = cl \left( \bigcup_n  \tilde{\Gamma}_{tr,a}^n (Q,R) \right)
\label{eq.indist_a}
\end{equation}
for which, thanks to  (\ref{eq.QeqR}), we have $\tilde{\Gamma}_{tr,a}^{\infty} (Q,R) \subseteq \tilde{\Gamma}_{tr,a}^{\infty} (Q,Q) = \Gamma_{tr,c}^{\infty} (Q)$.

We are now ready to prove our main result regarding the $HT_{tr,a}^{lr}$ game.
\begin{theorem}[Asymptotic payoff of the $HT_{tr,a}^{lr}$ game]
   The error exponent of the payoff associated to the profile $(\Lambda_{tr}^{*,n}(P_{t_D^N}), \tilde{f}(\cdot, t_A^N))$ is lower (upper) bounded as follows
\begin{align}
       \label{eq.lowerbound}
       \tilde{\varepsilon}_{tr,a} & \ge \min_{Q, R\in \CC} \big\{c [\DD(Q||P_X) + \DD(R||P_X)] \\ \nonumber
       & + \min_{P \in \tilde{\Gamma}_{tr,a}^{\infty}(Q,R)} \DD(P||P_Y) )\big\},
\end{align}
\begin{equation}
\label{eq.upperbound}
       \tilde{\varepsilon}_{tr,a} \le \min_{Q \in \CC} \big[ 2c \cdot \DD(Q || P_X) + \min_{P
\in \tilde{\Gamma}_{tr,a}^{\infty} (Q, Q)} \DD(P || P_Y) \big].
\end{equation}
\label{theo.HTa}
\end{theorem}
\begin{proof}
The proof is similar to the proof of Theorem \ref{theo.SanovTRc}, with the noticeable difference that now the lower and upper bounds are different hence preventing us to derive a precise expression for the error exponent.
Let us start with the lower bound. By recalling the definition of the false negative error probability, for any $n$ we can write:
\begin{align}
\label{eq.up_bound_Pfn_a_v}
P_{fn} & = \hspace{0.2cm}\sum_{t_D^N}  \sum_{t_A^N} P_X(t_D^N) P_X(t_A^N)  \hspace{-0.5cm }\sum_{P \in \tilde{\Gamma}_{tr,a}^n (P_{t_D^N}, P_{t_A^N})}  \hspace{-0.3cm} P_Y(T(P)) \nonumber\\
&  = \sum_{Q \in \PP_N}  \sum_{R \in \PP_N} \hspace{-0.1cm} P_X(T(Q))   P_X(T(R)) \hspace{-0.3cm}\sum_{P \in \tilde{\Gamma}_{tr,a}^n(Q,R)}  \hspace{-0.3cm}  P_Y(T(P))\nonumber\\
& \le \sum_{Q \in \PP_N} \sum_{R \in \PP_N} P_X(T(Q)) P_X(T(R)) \nonumber \\
& \quad \cdot (n+1)^{|\XX|} 2^{-n \hspace{-0.2cm} \min\limits_{P \in \tilde{\Gamma}_{tr,a}^n(Q,R)} \hspace{-0.2cm}\DD(P || P_Y)} \nonumber\\
& \le \sum_{Q \in \PP_N} P_X(T(Q)) (n+1)^{|\XX|} (N+1)^{|\XX|} \nonumber \\
& \quad \cdot 2^{-n \min\limits_{R \in \PP_N} [c \DD(R||P_X) +  \min\limits_{P \in \tilde{\Gamma}_{tr,a}^n(Q,R)} \hspace{-0.2cm}\DD(P || P_Y)]} \nonumber\\
& \le  (n+1)^{|\XX|}(N+1)^{2 |\XX|} \nonumber \\
& \quad \cdot 2^{-n \min\limits_{Q,R \in \CC} \large[ c \DD(Q||P_X) + c \DD(R||P_X) + \min\limits_{P \in \tilde{\Gamma}_{tr,a}^n(Q,R)} \hspace{-0.1cm} \DD(P || P_Y))\large]},
\end{align}
where in the last inequality we replaced the minimization over all $Q$ and $R$ in $\PP_N$, with a minimization over the entire space of pmf's. By taking the logarithm of both sides and by letting $n$ tend to infinity, the lower bound in equation (\ref{eq.lowerbound}) is proved.

We now turn the attention to the upper bound. To do so, let $Q^*$ be the pmf achieving the minimum in equation (\ref{eq.upperbound}). Due to the density of rational numbers within real numbers, we can find two sequences of pmf's $Q_n$ and $R_n$ that tend to $Q^*$ when $n$ tends to infinity, and such that $Q_n \in \PP_n$, $R_n \in \PP_n, \forall n$. By remembering that $N = nc$, we can say that the subsequences $Q_N = Q_{nc}$ and $R_N = R_{nc}$ also tend to $Q^*$ when $n$ (and hence $N$) tends to infinity. We can, then, use the subsequences $Q_N$ and $R_N$ to write the following chain of inequalities:
\begin{align}
\label{eq.low_bound_Pfn_a_v}
P_{fn} &  = \sum_{Q \in \PP_N}  \sum_{R \in \PP_N} \hspace{-0.1cm} P_X(T(Q))   P_X(T(R)) \hspace{-0.3cm}\sum_{P \in \tilde{\Gamma}_{tr,a}^n(Q,R)}  \hspace{-0.3cm}  P_Y(T(P))\nonumber\\
& \ge \sum_{Q \in \PP_N} \sum_{R \in \PP_N} \hspace{-0.1cm}  P_X(T(Q))  P_X(T(R)) \hspace{-0.5cm}\sum_{P \in \tilde{\Gamma}_{tr,a}^n(Q,R)} \hspace{-0.2cm} \frac{2^{-n\DD(P || P_Y)}}{(n+1)^{|\XX|}} \nonumber\\
& \stackrel{(a)}\ge \sum_{Q \in \PP_N} \sum_{R \in \PP_N} \hspace{-0.1cm} P_X(T(Q))  P_X(T(R)) \frac{2^{-n \hspace{-0.3cm} \min\limits_{P \in \tilde{\Gamma}_{tr,a}^n(Q,R)} \hspace{-0.3cm} \DD(P || P_Y)}}{(n+1)^{|\XX|}} \nonumber\\
& \ge \sum_{Q \in \PP_N} \sum_{R \in \PP_N} \hspace{-0.1cm} P_X(T(Q))\frac{2^{-n[c \DD(R||P_X)+ \hspace{-0.2cm} \min\limits_{P \in \tilde{\Gamma}_{tr,a}^n(Q,R)} \hspace{-0.3cm}  \DD(P || P_Y)]}}{(N+1)^{|\XX|}(n+1)^{|\XX|}} \nonumber\\
& \stackrel{(b)}  \ge \sum_{Q \in \PP_N} \hspace{-0.1cm} P_X(T(Q))\frac{2^{-n[c \DD(R_N||P_X)+ \hspace{-0.2cm} \min\limits_{P \in \tilde{\Gamma}_{tr,a}^n(Q,R_N)} \hspace{-0.3cm}  \DD(P || P_Y)]}}{(N+1)^{|\XX|}(n+1)^{|\XX|}} \nonumber\\
& \ge \sum_{Q \in \PP_N} \hspace{-0.1cm} \frac{2^{-n[c\DD(Q||P_X) + c\DD(R_N||P_X)+ \hspace{-0.2cm} \min\limits_{P \in \tilde{\Gamma}_{tr,a}^n(Q,R_N)} \hspace{-0.3cm}  \DD(P || P_Y)]}}{(N+1)^{2|\XX|}(n+1)^{|\XX|}} \nonumber\\
& \stackrel{(c)} \ge \frac{2^{-n[c\DD(Q_N||P_X) + c\DD(R_N||P_X)+ \hspace{-0.2cm} \min\limits_{P \in \tilde{\Gamma}_{tr,a}^n(Q_N,R_N)} \hspace{-0.3cm}  \DD(P || P_Y)]}}{(N+1)^{2|\XX|}(n+1)^{|\XX|}},
\end{align}
where inequalities $(a)$, $(b)$ and $(c)$ have been obtained by replacing the summation with a single element of the sum (two elements of the sequences $Q_N$ and $R_N$ for $(b)$ and $(c)$), and the others rely on a known lower bound on the probability of a type class (\cite{CandT}, chapter 12). By taking the logarithm of each side in (\ref{eq.low_bound_Pfn_a_v}), we can write:
\begin{align}
\label{eq.low_bound_errexp_a}
    \tilde{\varepsilon}_{tr,a} & \le c\DD(Q_N||P_X) + c\DD(R_N||P_X) \nonumber \\
    & + \min\limits_{P \in \tilde{\Gamma}_{tr,a}^n(Q_N,R_N)} \hspace{-0.3cm}  \DD(P || P_Y) + \beta_n,
\end{align}
with $\beta_n = [2 |\XX| \log (N+1) + |\XX| \log(n+1)]/n$ tending to 0 for $n \rarrow \infty$. Let then $P^*$ be defined as:
\begin{equation}
    P^* = \arg\min_{P \in \tilde{\Gamma}_{tr,a}^{\infty}(Q^*,Q^*)} \DD(P || P_Y).
\label{eq.Pstar_a}
\end{equation}
By recalling that both $Q_N$ and $R_N$ tend to $Q^*$ for increasing $N$, we can invoke Lemma \ref{Lemma_density_Gamma_a_n_in_var} in Appendix \ref{app.lemmaTondibis}, to build a sequence $P_n$ such that each term of the sequence belongs to $\tilde{\Gamma}_{tr,a}^{n}(Q_{N},R_N)$ and $P_n \rarrow P^*$, when $n \rarrow \infty$.
By recalling that
\begin{equation}
    Q^* = \arg\min_{Q \in \CC} \big[ 2c \cdot \DD(Q || P_X) + \min_{P \in \tilde{\Gamma}_{tr,a}^{\infty} (Q, Q)} \DD(P || P_Y) \big],
\label{eq.minQ}
\end{equation}
and by reasoning as in the proof of Theorem \ref{theo.SanovTRc} (equations (\ref{eq.up_bound_P_fn_3}) and (\ref{eq.fn_err_exp}) and discussion therein), we can eventually prove the upper bound (\ref{eq.upperbound}).
\end{proof}

Theorem \ref{theo.HTa} has an important corollary.
\begin{corollary}[Indistinguishability region for $HT_{tr,a}^{lr}$]
The false negative error exponent associated to the profile $(\Lambda_{tr}^{*,n}(P_{t_D^N}), \tilde{f}(\cdot, t_A^N))$ is equal to zero if and only if $P_Y \in \tilde{\Gamma}_{tr,a}^{\infty}(P_X, P_X) = \Gamma_{tr,c}^{\infty}(P_X)$, and hence the indistinguishability region of the $HT_{tr,a}^{lr}$ game is equal to that of the $HT_{tr,c}^{lr}$ game.
\label{cor.HTa}
\end{corollary}
\begin{proof}
From the upper bound in Theorem \ref{theo.HTa}, it follows that $\tilde{\varepsilon}_{tr,a} = 0$ if $P_Y \in \tilde{\Gamma}_{tr,a}^{\infty}(P_X, P_X)$, whereas from the lower bound we see that $\tilde{\varepsilon}_{tr,a} = 0$ implies that $P_Y \in \tilde{\Gamma}_{tr,a}^{\infty}(P_X, P_X)$.
\end{proof}

Corollary \ref{cor.HTa} provides an interesting insight into the achievable performance of the $HT_{tr,a}^{lr}$ game. While, in general, version $a$ of the game is less favorable to the attacker than version $c$, since in the latter case the attacker knows exactly the acceptance region adopted by the defender, if the attacker adopts the strategy $\tilde{f}$, the indistinguishability regions of the two games are the same. Such a strategy, then, is optimal at least as far as the indistinguishability region is concerned. On the other side, there is no guarantee that the attacker can achieve the same payoff as for version $c$.

\subsection{Training sequences with different length}

We conclude this section by discussing briefly the case in which the training sequences $t_D^N$ and $t_A^K$ have different lengths, i.e. $c \ne d$. To simplify the analysis we assume that $c$ is known to the attacker, in this way A knows at least the exact form the $h_c$ function used by D. We focus on the following attacking strategy: use the training sequence $t_A^K$ to estimate $P_{t_D^N}$ and use the estimate to attack the sequence $y^n$. Specifically, the attacker may use the following estimate of $P_{t_D^N}$:
\begin{align}
\label{eq.estimate}
    & \tilde{P}_{t_D^N}(i) = \frac{1}{N} \bigg\lfloor P_{t_A^K}(i) \cdot N \bigg\rfloor \quad \forall i = 1 \dots |\XX| -1, \nonumber \\
    & \tilde{P}_{t_D^N}(|\XX|) = 1 - \sum_{i=1}^{|\XX|-1} \tilde{P}_{t_D^N}(i),
\end{align}
to implement the attacking function:
\begin{equation}
	\tilde{f}(y^n,t_A^K) = \arg \min_{z^n : d(z^n, y^n) \le nD} h_c(P_{z^n}, \tilde{P}_{t_D^N}).
\label{equation}
\end{equation}
With the above definitions, we can easily extend the analysis carried out for the case $c = d$ and obtain very similar results. Specifically, the upper bound in Theorem \ref{theo.HTa} can be rewritten as:
\begin{equation}
\label{eq.upperbound_bis}
       \tilde{\varepsilon}_{tr,a} \le \min_{Q \in \CC} \big[ (c+d) \cdot \DD(Q || P_X) + \min_{P
\in \tilde{\Gamma}_{tr,a}^{\infty} (Q, Q)} \DD(P || P_Y) \big],
\end{equation}
whose proof is practically identical to the proof of Theorem \ref{theo.HTa} and is omitted for sake of brevity. By observing that the performance achievable by the defender in version $a$ of the game are at least as good as those achievable in version $c$, since in the latter case A knows exactly the acceptance region adopted by D and hence his attacks will surely be more effective, equation (\ref{eq.upperbound_bis}) allows us to conclude that the indistinguishability region is equal to that obtained for the case $c = d$.

We end this section by considering briefly version $b$ of the game. In some sense, we can say that version $b$ is halfway between versions $a$ and $c$. Like in version $a$, the attacker does not have a perfect knowledge of the training sequence used by the defender and hence he must resort to an estimate of the true acceptance region. On the other hand, the situation is more favorable to the attacker with respect to version $a$ with $d < c$, since now D knows at least part of the training samples used by D. Given that versions $a$ and $c$ of the game have the same indistinguishability region, we can conclude that the indistinguishability region of version $b$ will also be the same.

\section{Concluding remarks}
\label{sec.conc}

The need to protect the cyberworld that surrounds us has spurred researchers to look for suitable countermeasures against the ever increasing number of attacks that every day are brought against the digital world we live in. In many cases, though, research has focused on specific security threats, each time by developing tailored solutions that can not be easily extended to other scenarios. It is no surprise, then, that similar solutions are re-invented several times, and that the same problems are faced with again and again by ignoring that satisfactory solutions have already been discovered in contiguous fields. Even worse, the lack of a unifying view does not permit to grasp the essence of the addressed problems and work out effective and general solutions to be applied with limited effort to different fields. Times are ripe to develop general tools and models that can be used to analyze very general classes of problems wherein the presence of an adversary aiming at system failure can not be neglected.

As a first attempt in this sense, we have introduced a framework to analyze the achievable performance of binary hypothesis testing in an adversarial setting, i.e. when an adversary is present with the explicit goal of degrading the performance of the test. We did so by casting the hypothesis testing problem into a game-theoretic framework. In this way, in fact, we have been able to define rigorously the goals and constraints of the two contenders, namely the analyst, a.k.a. the defender, and the adversary or attacker. More specifically we introduced several versions of the hypothesis testing game, by paying attention to distinguish between hypothesis testing with known sources and hypothesis testing with training data. Given that a problem very similar the former case has already been studied in \cite{BT13}, we then focused on hypothesis testing with training data (the $HT_{tr}$ game). From a more technical point of view, we derived the asymptotic equilibrium point of various versions of the game, and analyzed the achievable payoff at the equilibrium. In addition to shedding a new light on the achievable performance of hypothesis testing in an adversarial environment, the analysis we carried out has the merit to clearly show the potentiality of the use of game-theoretic concepts coupled with tools typical of information theory and statistics.

Several directions for future research can be pointed out. The extension of our analysis to multiple hypothesis testing and classification is one of the most promising research directions, together with the extension to games characterized by more than two players. The analysis of situations in which the players do not have a perfect knowledge of the strategies available to the other players, or even the payoff function, by modeling the problem as a game with imperfect knowledge is another interesting research direction. The investigation of sequential games, in which the players move repeatedly each time by exploiting the result of the previous round of the game also offers many interesting hints for a fruitful research.

From a more focused perspective, it would be interesting to study specific instances of the $HT_{tr}$ game, in which the sources belong to a specific class, e.g. binary sources. We predict that in this way simpler expressions of the equilibrium point, the achievable payoff and, most of all, the indistinguishability region could be obtained, thus permitting to get additional interesting insights. Finally, we mention the opportunity of extending the analysis to the case of continuous sources. While the general ideas would remain the same, passing from discrete to continuous sources does not seem to be a trivial step, since our analysis relied heavily on the method of types, whose extension to continuous sources, though not impossible, comes with a number of additional difficulties.

\section*{Acknowledgment}
This work was partially supported by the REWIND Project, funded by the Future and Emerging Technologies (FET) program within the 7FP of the EC, under grant 268478.

\bibliographystyle{IEEEtran}
\bibliography{HTGbiblio}

\numberwithin{equation}{section}
\appendix

\renewcommand{\theequation}{A\arabic{equation}}

\subsection{Proof of Lemma \ref{lemma.property_Dsum}}
\label{app.lemma1}
We start by remembering that for memoryless sources we have (\cite{CandT}, chapter 12):
\begin{equation}
    n \DD(P_{x^n} || P_X) = - \log(P_X(x^n)) - n H(P_{x^n}).
\label{eq.dif_vs_entropy}
\end{equation}
By applying the above property to the right-hand side of equation (\ref{eq.property_Dsum}), we obtain:
\begin{align}
\label{eq.app1}
    & n\DD(P_{x^n}||P_X) + N\DD(P_{t^N} || P_X) = \\ \nonumber
    & -n H(P_{x^n}) - N H(P_{t^N}) - \log P_X(r^{n+N}),
\end{align}
where we have used the memoryless nature of $P_X$ due to which $P_X(r^{n+N}) = P_X(t^N) \cdot P_X(x^n)$. For any $P_X \in \CC$, we also have\footnote{Relationship (\ref{eq.app2}) can be proved, for instance, by exploiting the inequality $\ln x \ge (1-1/x)$ (where the equality holds if and only if $x = 1$).}:
\begin{equation}
    P_X(r^{n+N}) \le \prod_{a \in \XX} P_{r^{n+N}}(a)^{N_{r^{n+N}}(a)},
\label{eq.app2}
\end{equation}
where $N_{r^{n+N}}(a)$ indicates the number of times that symbol $a$ appears in $r^{n+N}$, and where equality holds if and only if $P_X(a) = P_{r^{n+N}}(a)$ for all $a$. By applying the log function we have:
\begin{align}
\label{app3}
    \log P_X(r^{n+N}) \le & \log \prod_{a \in \XX} P_{r^{n+N}}(a)^{N_{r^{n+N}}(a)} \\ \nonumber
    = &  \log \prod_{a \in \XX} P_{r^{n+N}}(a)^{(N_{x^n}(a) + N_{t^N}(a))} \\ \nonumber
    = & \sum_{a \in \XX} N_{x^n}(a) \log P_{r^{n+N}}(a) + \\ \nonumber
    ~ & \sum_{a \in \XX} N_{t^N}(a) \log P_{r^{n+N}}(a).
\end{align}
By inserting the above inequality in (\ref{eq.app1}), and by using the definition of empirical KL divergence, we obtain:
\begin{align}
\label{app4}
    & n\DD(P_{x^n}||P_X) + N\DD(P_{t^N} || P_X) \\ \nonumber
    ~ & \ge \sum_{a \in \XX} N_{x^n}(a) \log \frac{P_{x^n}(a)}{P_{r^{n+N}}(a)} + \sum_{a \in \XX} N_{t^N}(a) \log \frac{P_{t^N}(a)}{P_{r^{n+N}}(a)} \\ \nonumber
    ~ & = n \DD(P_{x^n} || P_{r^{n+N}}) + N \DD(P_{t^N} || P_{r^{n+N}}),
\end{align}
where the equality holds if and only if $P_X = P_{r^{n+N}}$, thus completing the proof.

\subsection{Topology of $\Gamma_{tr,c}^{\infty}(Q)$}
\label{app.lemmaTondi}
Many of the proofs in the body of the paper relies on the following lemma.
\begin{lemma}
\label{Lemma_density_Gamma_n_in_var}
Let $\{Q_{N(n)}\}$ be a sequence of pmf's such that $Q_{N(n)} \rarrow Q$ when $n \rarrow \infty$. Then, for any pmf $P^*$ in $\Gamma_{tr,c}^{\infty}(Q)$ a sequence $P_n$ exists such that $P_n \in \Gamma_{tr,c}^{n}(Q_{N(n)})$ for each $n$ and $P_n \rarrow P^*$.
\end{lemma}
\begin{proof}
To prove the lemma, we will show that for any $\varepsilon >0$ and $n$ large enough, we can find a pmf $P_n \in \Gamma_{tr,c}^n(Q_{N(n)})$ such that $d(P_n, P^*) < \varepsilon$. Specifically, we will do so by assuming that $P^* \in \cup_n \Gamma_{tr,c}^n(Q)$. If this is not the case, in fact, by the definition of $\Gamma_{tr,c}^{\infty}(Q)$, it is possible to find a pmf $P' \in \cup_n \Gamma_{tr,c}^n(Q)$ that is arbitrarily close to $P^*$, and then a pmf in $\Gamma_{tr,c}^n(Q_{N(n)})$ that is arbitrarily close to $P'$ and hence to $P^*$. Let then $P^*$ belong to $\Gamma_{tr,c}^m(Q)$ for some $m$. This means that for any sequence $y^m \in T(P^*)$ a mapping $f$ exists that transforms $y^m$ into a sequence $z^m$ such that
\begin{equation}
    h_c(P_{z^m},Q) = \lambda - \delta_m -\delta
\label{eq.appendix_1}
\end{equation}
with $\delta_m = |\XX| [\log(m+1)(N(m)+1)]/m$ and where $\delta$ is a strictly positive quantity. Due to the density of $\cup_n \PP_n$ in set of all pmf's, we can  find a sequence of pmf's $P_n$ that tends to $P^*$ when $n$ tends to infinity and for which $P_n \in \PP_n$ for each $n$. Let $y^n$ be a sequence in $T(P_n)$. Let $n_f(i \rarrow j)$ indicate the number of times that the mapping $f$ transforms the $i$-th symbol of the alphabet into the $j$-th one. By starting from the mapping $f$, for each $n$ we build a mapping $v^{(n)}$ for which $n_{v^{(n)}}(i \rarrow j) = \lfloor n_f(i \rarrow j) \cdot n/m \rfloor$. When $n$ increases the type of $z^n = v^{(n)}(y^n)$ will approach that of $z^m = f(y^m)$ for any $y^m$ in $P^*$. By remembering that $Q_{N(n)} \rarrow Q$ when $n$ tends to infinity, and by exploiting the continuity of the $h$ function, we can write:
\begin{eqnarray}
    h_c(P_{z^n}, Q_{N(n)}) & = & h_c(P_{z^n}, Q) + \beta'_n  \\ \nonumber
    & = & h_c(P_{z^m}, Q) + \beta'_n + \beta''_n \\ \nonumber
    & = & \lambda - \delta_m - \delta + \beta'_n + \beta''_n,
\label{eq.appendix_2}
\end{eqnarray}
where $\beta'_n$ and $\beta''_n$ tend to zero when $n \rarrow \infty$. Given that $\delta$ is a fixed and strictly positive number, we conclude that when $n$ is large, $P_{z^n}$ belongs to $\Lambda_{tr}^n(Q_{N(n)})$ and hence $P_n$ belongs to $\Gamma_{tr,c}^n(Q_{N(n)})$, thus completing the proof.
\end{proof}

\subsection{Topology of $\tilde{\Gamma}_{tr,a}^{\infty}(Q,Q)$}
\label{app.lemmaTondibis}
The analysis of version $a$ of the hypothesis testing game needs that lemma \ref{Lemma_density_Gamma_n_in_var} is generalized as follows.
\begin{lemma}
\label{Lemma_density_Gamma_a_n_in_var}
Let $\{Q_{N(n)}\}$ and $\{R_{N(n)}\}$ be two sequences of pmf's such that $Q_{N(n)} \rarrow Q$ and $R_{N(n)} \rarrow Q$ as $n \rarrow \infty$. Then, for any pmf $P^*$ in $\tilde{\Gamma}_{tr,a}^{\infty}(Q,Q)$ a sequence $P_n \in \tilde{\Gamma}_{tr,a}^{n}(Q_{N(n)}, R_{N(n)})$ exists such that $P_n \rarrow P^*$.
\end{lemma}
\begin{proof}
Given the definition of $\tilde{\Gamma}_{tr,a}^{\infty}(Q,Q)$ as the closure of $\bigcup_n \tilde{\Gamma}_{tr,a}^{n}(Q,Q)$, we proceed as in Lemma \ref{Lemma_density_Gamma_n_in_var}, and limit our proof to the special case in which $P^* \in \bigcup_n \tilde{\Gamma}_{tr,a}^{n}(Q,Q)$. Let then $P^*$ belong to $\tilde{\Gamma}_{tr,a}^m(Q,Q)$ for some $m$. This means that for any sequence $y^m \in T(P^*)$ the function
\begin{equation}
\label{eq.app_ftilde}
	\tilde{f}(y^m,Q) = \arg \min_{z^m : d(z^m, y^m) \le mD} h_c(P_{z^m}, Q)
\end{equation}
maps $y^m$ into a sequence $z^m$ such that
\begin{equation}
\label{eq.appendixC_1}
    h_c(P_{z^m},Q) = \lambda - \delta_m -\delta,
\end{equation}
with $\delta_m = |\XX| [\log(m+1)(N(m)+1)]/m$ and where $\delta$ is a strictly positive quantity. Due to the density of rational numbers in the real line, we can  find a sequence of pmf's $P_n$ that tends to $P^*$ when $n$ tends to infinity and for which $P_n \in \PP_n$ for each $n$. Let $u^n$ be a sequence in $T(P_n)$. By reasoning as in Lemma \ref{Lemma_density_Gamma_n_in_var}, we can define a, distortion-limited, mapping $v$ for which $n_{v}(i \rarrow j) = \lfloor n_{\tilde{f}}(i \rarrow j) \cdot n/m \rfloor$. As we have shown in Lemma \ref{Lemma_density_Gamma_n_in_var}, when $n$ is large, $v$ brings the sequence $u^n$ into a sequence $w^n \in \Lambda_{tr}^{*,n}$, so that, for any $Q_{N(n)} \rarrow Q$ and $n$ large enough, we have:
\begin{equation}
\label{eq.appendixC_2}
    h_c(P_{w^n},Q_{N(n)}) = \lambda - \delta_m -\delta +\beta'_n,
\end{equation}
with $\beta'_n$ approaching zero when $n$ increases. This, however, is not enough to ensure that $P_n \in \tilde{\Gamma}_{tr,a}^n(Q_{N(n)}, R_{N(n)})$, since for  this to be the case the minimization (\ref{eq.app_ftilde}) has to be carried by using $R_{N(n)}$ instead of $Q_{N(n)}$, and so there is no guarantee that the optimum mapping function will be $v$ (or another mapping better than that). Despite this observation, we can exploit the fact that $R_{N(n)}$ tends to $Q$ when $n$ tends to infinity, to show that indeed, for large $n$, $P_n \in \tilde{\Gamma}_{tr,a}^n(Q_N(n), R_{N(n)})$. Given $u^n \in P_n$, let $r^n$ be defined as follows:
\begin{equation}
\label{eq.app_ftildebis}
	r^n = \arg \min_{r^n : d(r^n, u^n) \le nD} h_c(P_{r^n}, R_{N(n)}).
\end{equation}
Since both $Q_{N(n)}$ and $R_{N(n)}$ tend to $Q$, when $n$ increases they will get arbitrarily close to each other. By exploiting the continuity of the $h_c$ function, we can write the following chain of inequalities:
%
%
\begin{align}
\label{eq.chainappendixC}
    h_c(P_{r^n},Q_{N(n)}) & \stackrel{(a)}\le h_c(P_{r^n}, R_{N(n)}) + \beta''_n \\ \nonumber
    & \stackrel{(b)}\le h_c(P_{w^n}, R_{N(n)}) + \beta''_n \\ \nonumber
    & \stackrel{(c)}\le h_c(P_{w^n}, Q_{N(n)}) + \beta''_n + \beta'''_n \\ \nonumber
    &\stackrel{(d)} = \lambda - \delta_m - \delta + \beta'_n + \beta''_n + \beta''_n \\ \nonumber
    & \stackrel{(e)}\le \lambda - \delta_n -\delta + \beta'_n + \beta''_n + \beta'''_n,
\end{align}
where $\beta'_n, \beta''_n$ and $\beta'''_n$ can be made arbitrarily small by increasing $n$, and where $(a)$ and $(c)$ derive from the continuity of the $h_c$ function and from the fact that $Q_{N(n)}$ and $R_{N(n)}$ tend to the same limit, $(b)$ is due to the fact that $r^n$ is the solution of the minimization in (\ref{eq.app_ftildebis}), $(d)$ derives from (\ref{eq.appendixC_2}), and $(e)$ is due to the fact that for a fixed $m$ when $n$ increases $\delta_n = |\XX| [\log(n+1)(N(n)+1)]/n \le  |\XX| [\log(m+1)(N(m)+1)]/m = \delta_m$.

Equation (\ref{eq.chainappendixC}) proves that $P_n \in \tilde{\Gamma}_{tr,a}^n(Q_N(n), R_{N(n)})$, thus completing the proof of the lemma.
\end{proof}

\end{document}